\documentclass[11pt]{scrartcl}

\usepackage[utf8]{inputenc}
\usepackage[english]{babel}
\usepackage{lmodern}
\usepackage{microtype}
\usepackage{amsmath}
\usepackage{amsthm}
\usepackage{amssymb}
\usepackage{mathtools}
\usepackage[numbers]{natbib}
\usepackage[boxed]{algorithm2e}
\usepackage{aliascnt}
\usepackage[colorlinks, pagebackref]{hyperref}
\usepackage{paralist}
\usepackage{graphicx}
\usepackage{booktabs}
\usepackage{threeparttable}
\usepackage{multirow}
\usepackage{tikz}
\usetikzlibrary{decorations.pathreplacing}
\usepackage[normal, small]{caption}
\usepackage{subcaption}
\usepackage{authblk}

\newcommand{\Rpos}{\mathbb{Q}^+_0}
\newcommand{\NP}{\textnormal{NP}}
\newcommand{\coNPpoly}{\textnormal{coNP/poly}}

\newcommand{\problemdef}[3]{
  \setlength{\pltopsep}{2ex}
  \begin{compactdesc}
    \item[\normalfont\textsc{#1}]
    \item[Input:] #2
    \item[Question:] #3
  \end{compactdesc}
  \setlength{\pltopsep}{0ex}
}
\newcommand{\DV}{\textsc{Dis\-tinct Vec\-tors}}
\newcommand{\hhDV}[2]{\textsc{Bi\-na\-ry} $(#1,#2)$-\textsc{Dis\-tinct Vec\-tors}}

\newcommand{\HS}{\textsc{Hitting Set}}
\newcommand{\fHS}{\textsc{$H$-Hitting Set}}
\newcommand{\SC}{\textsc{Set Cover}}

\newcommand{\DIS}{\textsc{Distance-$3$ Independent Set}}
\newcommand{\IM}{\textsc{Induced Matching}}
\newcommand{\IS}{\textsc{Independent Set}}

\newcommand{\N}{\mathbb{N}}
\newcommand{\W}[1]{\textnormal{W[#1]}}

\newcommand{\FPT}{\textnormal{FPT}}

\newcommand{\mbf}[1]{\mathbf{#1}}

\newcommand{\Scol}[2][S]{\ensuremath{#1_{|#2}}}
\newcommand{\row}[2]{\ensuremath{#1_{#2}}}
\newcommand{\col}[2]{\ensuremath{#1_{\ast #2}}}
\newcommand{\Ham}{\ensuremath{\Delta}}
\newcommand{\setsys}[1]{\ensuremath{\mathcal{W}_{#1}}}
\newcommand{\pet}[2]{\ensuremath{\widetilde{#1}_{#2}}}

\makeatletter
\def\NAT@spacechar{~}
\makeatother

\newtheorem{theorem}{Theorem}

\newaliascnt{lemma}{theorem}  
\newtheorem{lemma}[lemma]{Lemma}  
\aliascntresetthe{lemma}  

\newaliascnt{corollary}{theorem}
\newtheorem{corollary}[corollary]{Corollary}
\aliascntresetthe{corollary}

\newaliascnt{observation}{theorem}  
  
\aliascntresetthe{observation}

\newtheorem{rrule}{Reduction Rule}
\newtheorem{definition}{Definition}

\title{Exploiting Hidden Structure in Selecting Dimensions that Distinguish Vectors\footnote{A preliminary version appeared under the title ``A
  Parameterized Complexity Analysis of Combinatorial Feature Selection
  Problems'' in the proceedings of
the 38th International Symposium on Mathematical Foundations of
Computer Science (MFCS '13), volume 8087 of Lecture Notes in Computer
Science, pages 445--456, Springer, 2013 \cite{FBNS13}.
Parts of this work originate from the first author's master's thesis
on combinatorial feature selection \cite{froese12}.
This article now exclusively
focuses on the \DV{} problem and provides all proofs
in full detail. It additionally contains a new main result for \DV{}
regarding a computational complexity dichotomy for the parameters minimum and maximum pairwise Hamming distance of the data points.}}
\author[1]{Vincent Froese\thanks{Supported by Deutsche Forschungsgemeinschaft, project DAMM (NI 369/13).}}
\affil[1]{Institut f\"ur Softwaretechnik und Theoretische Informatik,
  TU Berlin, Germany, \texttt{\{vincent.froese, rolf.niedermeier, manuel.sorge\}@tu-berlin.de}}

\author[2]{René van Bevern\thanks{Supported by Deutsche Forschungsgemeinschaft, project DAPA (NI 369/12).}}
\affil[2]{Novosibirsk State University, Novosibirsk, Russia, \texttt{rvb@nsu.ru}}

\author[1]{Rolf Niedermeier}

\author[1]{Manuel Sorge\protect\footnotemark[3]}

\date{}

\begin{document}
\maketitle
\renewcommand{\sectionautorefname}{Section}
\renewcommand{\subsectionautorefname}{Section}

\begin{abstract}

The NP-hard \DV{} problem asks to delete as many
columns as possible from a matrix such that all rows in the
resulting matrix are still pairwise distinct.
Our main result is that, for binary matrices, there is a complexity
dichotomy for \DV{} based on the maximum~($H$) and the
minimum~($h$) pairwise Hamming distance between matrix rows:
\DV{} can be solved in polynomial time if $H\leq 2\lceil h/2\rceil +1$, and
is NP-complete otherwise.
Moreover, we explore connections of \DV{} to hitting sets,
thereby providing several fixed-parameter tractability and
intractability results also for general matrices.
\end{abstract}

% \begin{keyword}
%   NP-Hardness, Fixed-Parameter Tractability, W-Hardness, Machine
%   Learning, Combinatorial Feature Selection, Dimension Reduction, Minimal Reduct Problem,
%   Combinatorics of Binary Matrices, $\Delta$-Systems
% \end{keyword}

\section{Introduction}
Feature selection in a high-dimensional feature space means to choose a
subset of features (that is, dimensions) such that some desirable
data properties are preserved or achieved in the induced subspace.
\emph{Combinatorial} feature selection~\citep{KS96,charikar00} is a
well-motivated alternative to the more frequently studied affine feature
selection.  While \emph{affine} feature selection combines features to
reduce dimensionality, combinatorial feature selection simply discards
some features. The advantage of the latter is that
the resulting reduced feature space is easier to interpret. See
\citet{charikar00} for a more extensive discussion in favor of
combinatorial feature selection.
Unfortunately, combinatorial feature selection problems are typically
computationally very hard to solve (NP-hard and also hard to
approximate~\citep{charikar00}), resulting in the use of heuristic
approaches in practice~\citep{BL97,DDHJM07,For03,GE03}.

In this work, we adopt the fresh perspective
of parameterized complexity analysis. We thus refine the known picture
of the computational complexity landscape of a prominent and formally
simple combinatorial feature selection problem called \DV{}.

  \problemdef{\DV}
    {A matrix $S \in \Sigma^{n\times d}$ over a finite alphabet~$\Sigma$ with~$n$ distinct rows and $k\in\N$.}
    {Is there a subset~$K\subseteq [d]$ of column indices with~$|K|\le
      k$ such that all~$n$ rows in~$\Scol{K}$ are still distinct?}
Here,~$\Scol{K}$ is the submatrix containing only the columns with
indices in~$K$.
In the above formulation, the input data is considered to be a matrix
where the row vectors correspond to the data points and the columns
represent features (dimensions).
Thus, \DV{} constitutes the basic task to compress the data by discarding
redundant or negligible dimensions without losing the essential
information to tell apart all data points.

Intuitively speaking, the guiding principle of this work is to
identify problem-specific parameters (quantities such as the number of
dimensions to discard or the number of dimensions to keep) and to analyze
how these quantities influence the computational complexity of \DV{}.
The point here is that in relevant applications these parameters can
be small, which may allow for more efficient solvability.
Hence, the central question is whether \DV{} is computationally tractable in the case of small parameter values.

We are particularly interested in the complexity of \DV{} if the
range of differences between data points is small. This special case
occurs if the input data is in some sense homogeneous. We measure the range of differences as the gap~$H - h$ between the maximum~$H$ and the minimum~$h$ of pairwise Hamming distances of rows in the input matrix.\footnote{See \autoref{sec:DVbinary} for a formal definition.} We initiate the study of this measure by completely classifying the classical complexity of \DV{} with respect to constant values of~$H - h$ on binary input matrices. For general matrices, we derive various tractability and intractability results with respect to the parameters alphabet size~$|\Sigma|$, number of retained columns and number of discarded columns.

% This special case could arise, for example, with \emph{sparse}
% data sets where the row vectors mainly contain 0's.
% In this situation, it may hold that any pair of rows differs in at
% most~$\beta$ columns.

\paragraph*{Related Work}

\DV{} is also known as the \textsc{Minimal Reduct} problem in rough
set theory~\citep{pawlak91}
and it was already early proven to be \NP-hard by~\citet{skowron92}.
Later, \citet{charikar00} investigated the computational complexity of
several problems arising in the context of combinatorial feature
selection, including \DV{}.
Seemingly unaware of Skowron and Rauszer's work, they showed that
there exists a constant~$c$ such that it is NP-hard to approximate
\DV{} in polynomial time within a factor of $c\log d$.

Another combinatorial feature selection problem called
\textsc{Minimum Feature Set} is a variant of \DV{} where not all pairs of rows have to be
distinguished but only all pairs of rows from two specified subsets.
This problem is known to be NP-complete for binary input data \cite{davies94}.
In addition, \citet{cotta03} investigated the parameterized complexity of
\textsc{Minimum Feature Set} and proved \W{2}-completeness with respect to the
number of selected columns even for binary matrices.

\paragraph*{Results and Outline}
\autoref{tab:results} summarizes our results. We first focus on the
case of input matrices over binary alphabets, that is~$|\Sigma| = 2$,
in \autoref{sec:DVbinary}. As our main result, we completely classify the classical computational
complexity of (binary) \DV{} according to the gap between~$H$ and~$h$.
This yields the following dichotomy:
If~$H \le 2\lceil h/2\rceil+1$, then \DV{} is polynomial-time
solvable, whereas it is \NP-complete in all other cases.
The corresponding \NP-completeness proof also implies
\W{1}-hardness with respect to the parameter ``number~$t=d-k$ of columns to
discard''.

In \autoref{sec:DVgeneral} we consider general alphabets, that is, $|\Sigma| \geq 2$. We prove that, here, \DV{} is \W{2}-hard with respect to the number~$k$ of
retained columns if the alphabet size is unbounded. Moreover, \DV{} cannot be solved in~$d^{o(k)}(nd)^{O(1)}$ time, unless $\W{1}=\FPT$ (which is strongly believed not
to be the case~\cite{downey13}). In contrast to these hardness results, we develop polynomial-time data reduction algorithms and show fixed-parameter tractability by providing superexponential-size problem kernelizations with respect to the
combined parameters $(|\Sigma|,k)$ and~$(H,k)$. We also exclude polynomial-size problem kernels with respect to the parameter combination $(n,|\Sigma|,k)$ based on the hypothesis that $\text{NP}\not\subseteq\text{coNP/poly}$ (which is believed to be true, since otherwise the polynomial hierarchy collapses to its third level). Finally, as a simple observation, we also give a linear-time factor-$H$ approximation algorithm.

Our notation is explained in \autoref{sec:prelim}. \autoref{sec:conc} concludes with some challenges for future research.

\begin{table}[t]
  \centering
  \begin{threeparttable}
    \caption{Overview of our results.}
    \label{tab:results}
    \begin{tabular}{ll}
      \toprule
      Result\tnote{*} & Reference\\
      \midrule
      \NP-hard for $|\Sigma|=2$ and $H\ge
      2\lceil h/2\rceil+2$ & \autoref{thm:abDV_NP-hard}\\
      poly-time for $|\Sigma|=2$ and $H\le 2\lceil
      h/2\rceil+1$ & \autoref{thm:hhDV-poly}\\
      \W{1}-hard wrt.~$t$ for $|\Sigma|=2$ and $H\geq 4$ & \autoref{cor:distinct_vectors_w1-hard_t}\\
      \W{2}-hard wrt.~$k$ ($|\Sigma|$ unbounded) & \autoref{thm:distinct_vectors_w2-hard}\\
      \FPT{} wrt.~$(|\Sigma|,k)$ (no poly kernel wrt.~$(n,|\Sigma|,k)$) & \autoref{thm:dv_kernel_sk}\\
      \FPT{} wrt.~$(H,k)$ (for arbitrary $\Sigma$) & \autoref{thm:dv_kernel_hk}\\
      \bottomrule
    \end{tabular}
    \begin{tablenotes}
      \footnotesize
    \item[*] $|\Sigma|$: alphabet size, $h$ ($H$): minimum (maximum)
      pairwise row Hamming distance of the input matrix,
      $t$: number of discarded columns, $k$: number of retained columns
    \end{tablenotes}
  \end{threeparttable}
\end{table}

% \paragraph*{Organization} In \autoref{sec:prelim}, we introduce our
% notation and basic concepts from parameterized complexity theory.
% \autoref{sec:DVbinary} constitutes the main part of this work discussing
% the results for \hhDV{\alpha}{\beta} followed by the results for \DV{} on
% general matrices in \autoref{sec:DVgeneral}.
% In \autoref{sec:conc}, we provide a brief conclusion.

\section{Preliminaries}
\label{sec:prelim}

\paragraph*{Notation}
For~$n\in\N$, let~$[n]:=\{1,\ldots,n\}$. The set of all size-$k$
subsets of a set~$X$ is denoted by~$\binom{X}{k}$.
In the following, we consider finite alphabets~$\Sigma\subseteq\Rpos$.
We denote by~$S=(s_{ij})\in\Sigma^{n\times d}$ the matrix with~$n$
rows and~$d$ columns, where $s_{ij}\in\Sigma$ denotes the entry in the
$i$-th row and the~$j$-th column.
We denote the~$i$\nobreakdash-th row vector by~$\row{s}{i}$ and the $j$-th
column vector by~$\col{s}{j}$. For subsets~$I\subseteq [n]$ and $J\subseteq [d]$ of row and column
indices, we write~$S[I,J]:=(s_{ij})_{(i,j)\in I\times J}$ for the~$|I|\times |J|$
submatrix of~$S$ containing only the rows with indices in~$I$ and the
columns with indices in~$J$. We use the
abbreviation~$\Scol{J}:=S[[n],J]$ for the submatrix containing all rows
but only the columns in~$J$ and we say that the columns
in~$[d]\setminus J$ are \emph{discarded} (or deleted).
For a vector~$x\in\Sigma^d$, we denote by~$(x)_j\in\Sigma$ the $j$-th
entry of~$x$. The null vector is denoted by~$\mbf{0}:=(0,\ldots,0)$.

Throughout this work, we assume that arithmetic operations such as
additions and comparisons of numbers can be done in~$O(1)$ time
(that is, we use the RAM model \cite{Pap94}).

\paragraph*{Parameterized Complexity}
We assume the reader to be familiar with the basic concepts from classical
complexity theory, such as NP-hardness and polynomial-time reductions
\cite{garey79, Pap94}.
The computational complexity of a parameterized problem
is measured in terms of two quantities: one is the input size, the
other is the \emph{parameter} (usually a positive integer).  A
parameterized problem~$L\subseteq \Sigma^*\times \mathbb N$ is called
\emph{fixed-parameter tractable} with respect to a parameter~$k$ if it
can be solved in~$f(k)\cdot |I|^{O(1)}$ time, where~$f$ is a
computable function only depending on~$k$, and~$|I|$ is the size of the
input instance~$I$.
A \emph{problem kernel} for a parameterized problem~$P$ is a
polynomial-time self-reduction, that is, given an instance~$(I,k)$,
it outputs another instance~$(I',k')$ of~$P$
such that~$|I'|+k' \le g(k)$ for some computable function~$g$
depending only on~$k$, and~$(I,k)$ is a yes-instance of~$P$ if and
only if~$(I',k')$ is a yes-instance of~$P$.
The function~$g$ is called the size of the problem kernel. If~$g$
is a polynomial, then we speak of a polynomial kernel.
Existence of a problem kernel is equivalent to fixed-parameter
tractability~\citep{Cyg15,downey13,flum06,niedermeier06}. 

A \emph{parameterized reduction}
from a parameterized problem~$P$ to another parameterized problem~$P'$
is a function that, given an instance~$(I,k)$ of~$P$, computes in~$f(k)\cdot
|I|^{O(1)}$ time an instance~$(I',k')$ (with~$k'$ only depending on~$k$)
such that~$(I,k)$ is a yes-instance of~$P$ if and only if~$(I',k')$ is
a yes-instance of~$P'$.  The two basic complexity classes for showing
(presumable) fixed-parameter intractability are called~W[1] and~W[2];
there is good complexity-theoretic reason to believe that W[1]-hard
and W[2]-hard problems are not fixed-parameter
tractable~\citep{Cyg15,downey13,flum06,niedermeier06}.

\section{Binary Matrices and the Range of Differences}
\label{sec:DVbinary}
% In this section, we provide a thorough discussion of the dimension
% reduction problem \DV{} concerning its classical as well as its
% parameterized complexity comprising several new results in both
% areas.
% We thereby extend the results of \citet{charikar00} by
% analyzing various restricted scenarios for \DV{}, conducting a more fine-grained
% computational complexity analysis which, however, mostly yields
% further hardness results. In particular, we consider the cases of
% \begin{inparaenum}[(i)]
% \item \emph{retaining} few columns,\label{enu:dv-retdim}
% \item \emph{deleting} few columns, and\label{enu:dv-deldim}
% \item \emph{small} pairwise Hamming distance of the input vectors.\label{enu:dv-diff}
% \end{inparaenum}

  Throughout this section, we focus on instances with a binary
  input alphabet, say, without loss of generality, $\Sigma=\{0,1\}$.
  We analyze the computational complexity with respect to the range of differences between input data points. To this end, we consider instances where the
  Hamming distance of each pair of rows lies within a prespecified
  range. In other words, the number of columns in which a given pair of
  rows differs shall be bounded from below and above by some
  constants~$\alpha,\beta\in \N$. We first give the formal definitions and then completely classify the classical complexity of \DV{} with respect to the gap between $\alpha$ and~$\beta$. The NP-complete cases are given in \autoref{sec:binDVNPc} and the polynomial cases in \autoref{sec:binDVpoly}. The formal definitions for our setup
  are the following.

  \begin{definition}[Weight]
    \label{def:weight}
    For a vector~$x\in\{0,1\}^d$, we denote by~$W_x := \{j\in [d]\mid
    (x)_j = 1\}$ the set of indices where~$x$ equals~1 and we call~$w(x) :=
    |W_x|$ the \emph{weight} of~$x$.
  \end{definition}

  \begin{definition}[Hamming Distance]
    \label{def:hamming}
    For vectors~$x,y \in \Sigma^d$, let~$D_{xy}:=\{j\in [d]\mid
    (x)_j \neq (y)_j\}$ be the set of indices where~$x$ and~$y$ differ
    and let~$\Ham(x,y) := |D_{xy}|$ denote the \emph{Hamming distance}
    of~$x$ and~$y$.    
  \end{definition}

  Note that, for~$x,y\in\{0,1\}^d$, it holds~$D_{xy}=(W_x\cup
  W_y)\setminus (W_x\cap W_y)$ and thus~$\Ham(x,y) = w(x) + w(y)
  - 2|W_x\cap W_y|$. For a \DV{} instance~$(S\in\Sigma^{n\times d},k)$,
we define the parameters \emph{minimum pairwise row Hamming distance}~$h:=\min_{i\neq
  j\in[n]}\Ham(\row{s}{i},\row{s}{j})$
and \emph{maximum pairwise row Hamming distance}~$H:=\max_{i\neq
  j\in[n]}\Ham(\row{s}{i},\row{s}{j})$.
To conveniently state our results, let us now define a variant of \DV{} with minimum pairwise row Hamming distance~$\alpha$ and maximum pairwise row Hamming distance~$\beta$:

 \problemdef{\hhDV{\alpha}{\beta}}
    {A matrix $S\in\{0,1\}^{n\times d}$ with~$n$ distinct rows such that
     $\alpha=h \le H = \beta$, and $k\in\N$.}
    {Is there a subset $K\subseteq [d]$ of column indices with
     $|K|\leq k$ such that all rows in $\Scol{K}$ are still distinct?}
  
  Intuitively, if the matrix consists of rows that are all ``similar'' to each other,
  one could hope to be able to solve the instance efficiently since
  there are at most~$\beta$ columns to choose from in order to
  distinguish two rows.
  The minimum pairwise row Hamming distance~$\alpha$ plays a dual role in
  the sense that, if~$\alpha$ is large, then each pair of rows differs
  in many columns, which also could make the instance easily solvable.
  The following theorems, however, show that this intuition is
  somewhat deceptive in the sense that \hhDV{\alpha}{\beta} is NP-hard
  even for small constants~$\alpha$ and $\beta$.
  Despite this intimidating NP-hardness result, we perform a closer inspection of the
  relation between the minimum and maximum pairwise row Hamming
  distance and show that it is possible to solve some cases in
  polynomial time for arbitrarily large constants~$\alpha$, $\beta$.
  These results are obtained by applying
  combinatorial arguments from extremal set theory
  revealing a certain structure of the input matrix if the values of~$\alpha$
  and~$\beta$ are close to each other, that is, the range of differences is small.
  Analyzing this structure, we can show how to find solutions in polynomial time.

  % \begin{table}[t]
  %   \renewcommand{\arraystretch}{1.25}
  %   \center
  %   \caption{Overview of the complexity of \hhDV{\alpha}{\beta} (+:
  %     P, $-$: NP-hard).}
  %   \label{tab:abDVresults}
  %   \begin{tabular}{cc|cccccccccc}
  %     \multicolumn{2}{c}{} & \multicolumn{10}{c}{$\beta$}\\
  %     & \multicolumn{1}{c}{} & 1 & 2 & 3 & 4 & 5 & 6 & 7 & 8 & 9 & $\cdots$\\
  %     \cline{3 - 12}
  %     \multirow{9}{*}{$\alpha$}
  %     & 1 & + & + & + & $-$ & $-$ & $-$ & $-$ & $-$ & $-$ & $\cdots$\\
  %     & 2 &   & + & + & $-$ & $-$ & $-$ & $-$ & $-$ & $-$ & $\cdots$\\
  %     & 3 &   &   & + & + & ? & $-$ & $-$ & $-$ & $-$ & $\cdots$\\
  %     & 4 &   &   &   & + & + & $-$ & $-$ & $-$ & $-$ & $\cdots$\\
  %     & 5 &   &   &   &   & + & + & ? & $-$ & $-$ & $\cdots$\\
  %     & 6 &   &   &   &   &   & + & + & $-$ & $-$ & $\cdots$\\
  %     & 7 &   &   &   &   &   &   & + & + & ? & $-$\\
  %     & 8 &   &   &   &   &   &   &   & + & + & $-$\\
  %     & 9 &   &   &   &   &   &   &   &   & + & +\\
  %     & $\vdots$ & & & & & & & & & & $\ddots$
  %   \end{tabular}
  % \end{table}

  \begin{figure}[t]
    \centering
    \begin{tikzpicture}[scale=.4]
      \foreach \x in {1,...,10} {
        \draw[help lines] (0,\x-1) rectangle (\x,\x);
        \draw[help lines] (\x-1,\x-1) rectangle (\x,10);
      }
      \draw[->] (0,0) -- (0,10);
      \draw[->] (0,0) -- (10,0);
      \node (a) at (10,-0.5) {$\alpha$};
      \node (b) at (-0.5,10) {$\beta$};
      % \foreach \x in {3,5,7} {
      %   \draw[fill=blue!30!white,opacity=.7] (\x-1,\x+1) rectangle (\x,\x+2);
      %   \node (?) at (\x-0.5,\x+1.5) {?};
      % }
      \filldraw[fill=gray!20!white,opacity=.7]
        (0,10) -- (0,3) -- (2,3) --
        (2,5) -- (4,5) -- (4,7) --
        (6,7) -- (6,9) -- (8,9) -- (8,10);
      \foreach \x in {1,...,9}{
        \node[font=\sffamily\scriptsize] (x) at (\x-0.5,-0.5) {\x};
        \node[font=\sffamily\scriptsize] (y) at (-0.5,\x-0.5) {\x};
      }
      \node (np) at (2.0,7.5) {NP-c.};
      \node (p) at (1.0,1.5) {$\in$ P};
    \end{tikzpicture}
    \caption{Overview of the complexity of \hhDV{\alpha}{\beta}. Gray cells correspond to NP-complete cases, whereas white cells are polynomial-time solvable cases.}
    \label{fig:abDVresults}
  \end{figure}
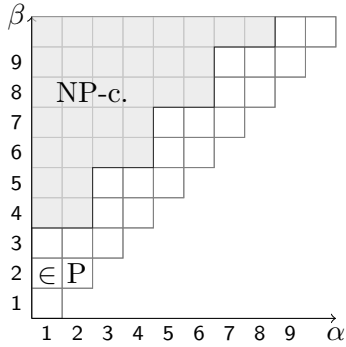
  
  \autoref{fig:abDVresults} depicts the (non-parameterized) computational
  complexity landscape for \hhDV{\alpha}{\beta} with respect
  to~$\alpha$ and~$\beta$, indicating the border of hardness.
  In the following, we will step by step develop the results exhibited
  in~\autoref{fig:abDVresults}.

  \subsection{NP-Completeness for Heterogeneous Data}\label{sec:binDVNPc}
  As a starting point, we prove the NP-completeness of the
  case~$\alpha=2$, $\beta=4$.

  \begin{theorem}
    \label{thm:24dv_NP-hard}
    \hhDV{2}{4} is NP-complete.
  \end{theorem}

\begin{proof}
    It is easy to check that \DV{} is in \NP.
    To prove \NP-hardness, we give a polynomial-time many-one
    reduction from a special variant of the
    \IS{} problem in graphs, which is defined as follows.

    \problemdef{\DIS}
    {An undirected graph $G=(V,E)$ and $k \in \mathbb{N}$.}
    {Is there a subset of vertices $I\subseteq V$ of size at least $k$ such that each pair of vertices from
      $I$ has distance at least three?}

    \noindent Here, the distance of two vertices is the number of
    edges contained in a shortest path between them.
    \DIS{} is known to be \NP-complete by a reduction from the
    \NP-complete \IM{} problem~\citep{brandstaedt11}.
  
    Our reduction works as follows: Let $(G=(V,E),k)$ with $|V|=n$ and
    $|E|=m$ be an instance of \DIS{} and let~$Z\in\{0,1\}^{m\times n}$ be the
    incidence matrix of~$G$ with rows corresponding to edges and
    columns to vertices, that is, $z_{ij}=1$~means that the $i$-th edge contains the~$j$-th vertex.
    We assume that~$G$ contains no isolated vertices since they are
    always contained in a maximum distance-3 independent set and can
    thus be removed. We further assume that~$G$ contains at least four
    edges of which at least two are disjoint.
    Otherwise, $G$ is either of constant size or a star, for which a
    maximum distance-3 independent set consists of only a single
    vertex. Hence, we can solve these cases in polynomial time and
    return a trivial yes- or no-instance.
    
    The matrix~$S\in\{0,1\}^{(m+1)\times n}$ of the \hhDV{2}{4} instance~$(S,k')$ is defined as follows:
    $\row{s}{i} := \row{z}{i}$ for all~$i\in[m]$ and $\row{s}{m+1}:=\mbf{0}$.
    The desired solution size is set to $k':=n-k$.
    Notice that each row in~$Z$ contains exactly two~1's and no two
    rows are equal since~$G$ contains no multiple edges.
    Moreover, by assumption, there exists a pair of rows with Hamming
    distance four since~$G$~contains a pair of disjoint edges.
    Since~$S$ contains the null vector as a row, it follows
    that~$h=2$ and~$H=4$.
    The instance~$(S,k')$ can be computed in~$O(nm)$ time.

    The correctness of the reduction is due to the following argument.
    The instance~$(G,k)$ is a yes-instance if and only if there is a set~$I \subseteq V$ of size exactly~$k$ such that every edge in~$G$ has at least one endpoint
    in~$V\setminus I$ and no vertex in~$V\setminus I$ has two
    neighbors in~$I$. 
    In other words, the latter condition says that no two edges with
    an endpoint in~$I$ share the same endpoint in~$V\setminus I$. 
    Equivalently, for the subset~$K$ of columns corresponding to the
    vertices in~$V\setminus I$, it holds that all rows
    in~$S[[m],K]$ contain at least one 1 and no two rows
    contain only a single~1 in the same column. 
    This is true if and only if~$K$~is a solution for~$(S,k')$ because
    $\row{s}{m+1}$ equals the null vector and thus two rows in~$\Scol{K}$
    can only be identical if either they consist of 0's only or
    contain only a single 1 in the same column. Furthermore, $|K| = |V \setminus I| = n - k = k'$.
  \end{proof}

  \noindent We remark that from a \W{1}-hardness result for
  \IM{} parameterized by the number of vertices in the induced
  subgraph~\citep{moser09},
  we can infer \W{1}-hardness for \DIS{} with respect to the solution size~$k$.
  Since the proof of \autoref{thm:24dv_NP-hard} actually provides a
  parameterized reduction from \DIS{} parameterized by~$k$ to \DV{}
  parameterized by the number of columns to discard (which
  is~$d-k'=n-(n-k)=k$ in the above reduction), we have the following:
  \begin{corollary}
    \label{cor:distinct_vectors_w1-hard_t}
    \hhDV{2}{4} is \W{1}-hard with respect to the
    number~$t:=d-k$ of discarded columns.
  \end{corollary}

  \noindent Note, however, that the reduction in the proof of \autoref{thm:24dv_NP-hard}
  is not a parameterized reduction with respect to the number~$k'=n-k$ of retained columns since~$k'$ does not solely depend on~$k$ but also on the number~$n$ of vertices.
  Hence, we cannot infer \W{1}-hardness with respect to~$k'$.
  In fact, we will show in \autoref{sec:DVgeneral} that \DV{} allows a problem kernel with respect to the number of retained columns for binary alphabets.
  
  We will now give polynomial-time reductions from \hhDV{2}{4} to
  certain other cases of \hhDV{\alpha}{\beta} with different bounds on the minimum
  and maximum Hamming distance.
  Using \autoref{thm:24dv_NP-hard} as an anchor point, we can then derive all remaining
  NP-completeness results in \autoref{fig:abDVresults}.
  The reductions will mainly build on some padding arguments, that is,
  starting from a given input matrix, we expand it by adding
  new columns and rows such that we achieve the desired
  constraints on the Hamming distances without changing the actual
  answer to the original instance.
  To start with, we define a type of column vectors which can
  be used for padding an input matrix without changing the answer
  to the original instance, that is, such ``padding columns'' are not contained in an
  optimal solution. Informally, a column~$j$ is
  \emph{inessential} if all rows could still be distinguished by the same
  number of columns without selecting~$j$.
  The formal definition is the following.

  \begin{definition}
    \label{def:redundant}
    For a matrix $S\in\Sigma^{n\times d}$, a column~$j\in [d]$ is
    called \emph{inessential} if the following two conditions are
    fulfilled:
    \begin{enumerate}[(1)]
      \item\label{cond1} There exists a row $i\in[n]$ such that column~$j$
        \emph{exactly distinguishes} row~$i$ from all other rows, that is, $s_{ij}\neq s_{lj}$
        and~$s_{lj} = s_{l'j}$ holds for
        all~$l,l'\in[n]\setminus\{i\}$.
      \item\label{cond2} All rows in $\Scol{[d]\setminus\{j\}}$ are still distinct.
    \end{enumerate}
  \end{definition}
  Note that for binary matrices, Condition~(\ref{cond1}) of
  \autoref{def:redundant} can only be fulfilled by column vectors
  that contain either a single~1 or a single~0, that is,
  the column vectors of weight~1 or~$n-1$.

  Next, we show that, for any inessential column
  in a given input matrix, we can assume that this column
  is not contained in a solution for the \DV{} instance.
  
\begin{lemma}
  \label{lem:padding}
  Let~$(S\in\{0,1\}^{n\times d},k)$ be a \DV{} instance
  with an inessential column~$j\in[d]$.
  It holds that~$(S,k)$ is a yes-instance if
  and only if~$(\Scol{[d]\setminus\{j\}},k)$ is a yes-instance.
\end{lemma}

\begin{proof}
  It is clear that the ``if'' part of the statement holds; let us consider
  the ``only if'' part. To this end, assume that there is a
  solution set~$K\subseteq [d]$ of columns for~$(S,k)$ with~$j\in K$.
  Since column~$j$ exactly distinguishes row~$i$ from all
  other rows and no other pair of rows, it follows
  that~$K':=K\setminus \{j\}$ is a solution for~$(S[[n]\setminus\{i\},[d]\setminus\{j\}],k-1)$.
  But then, there also exists a solution~$K''\subseteq [d]\setminus\{j\}$
  for~$(S[[n],[d]\setminus\{j\}],k)$. This is true because
  row~$i$ equals at most one other row~$l$ in~$S[[n],K']$ since all
  rows in~$S[[n]\setminus\{i\},K']$ are distinct.
  Row~$i$ can thus be distinguished from row~$l$ by a column~$j' \in [d]\setminus\{j\}$
  with~$s_{ij'}\neq s_{lj'}$, which exists because column~$j$ is
  inessential, and thus, by definition, all rows in~$S[[n],[d]\setminus\{j\}]$ are distinct.
  Hence, $K'' := K' \cup \{j'\}$ is a solution for~$(S,k)$.
\end{proof}

Note that, due to \autoref{lem:padding}, adding inessential columns
to a given input matrix yields an equivalent \DV{} instance.
Hence, for the binary case, any construction that only adds column vectors
which either contain a single~1 or a single~0 to the
input matrix yields an equivalent instance since these columns are
clearly inessential.
Following this basic idea, the proof of the following theorem
shows which Hamming distances can be generated from
a given input matrix by adding inessential columns.

\begin{theorem}
    \label{thm:abDV_NP-hard}
    % For all~$\alpha, \beta \ge 0$, it holds that
    % \begin{enumerate}[1.)]
    %   \item\label{thm:abDV_NP-hard_i} \hhDV{1}{4 + \beta} is
    %     NP-complete and
    %   \item\label{thm:abDV_NP-hard_ii} \hhDV{2 + \alpha}{4 + 2\lceil\alpha/2\rceil +
    %       \beta} is NP-complete.
    % \end{enumerate}
    \hhDV{\alpha}{\beta} is NP-complete for all
    \begin{itemize}
      \item $\beta \ge \alpha + 2$ if $\alpha$ is even, and
      \item $\beta \ge \alpha + 3$ if $\alpha$ is odd.
    \end{itemize}
  \end{theorem}

\begin{proof}
  In the following, we give polynomial-time many-one reductions from
  \hhDV{2}{4}.
  To this end, let~$(S\in\{0,1\}^{n\times d},k)$ be the
  \hhDV{2}{4} instance as constructed in the proof of
  \autoref{thm:24dv_NP-hard}. Recall that this matrix~$S$ contains the null row vector,
  say~$\row{s}{n}=\mbf{0}$, and all other rows have
  weight two, $w(\row{s}{i})=2$ for all~$i \in[n-1]$.
  Moreover, there exists a pair of rows with Hamming distance four.
  Assume, without loss of generality, that the first two
  rows~$\row{s}{1}$ and~$\row{s}{2}$ have Hamming
  distance~$\Ham(\row{s}{1},\row{s}{2})=4$.
  \autoref{fig:subS} depicts an example of such a matrix.
  In the following, let~$D_j = D_{\col{s}{j}} \subseteq [n]$ denote the
  set of row indices where column vector~$\col{s}{j}$ equals 1.
  Further, for~$i\in\N$ and $I\subseteq [i]$, let~$\mbf{1}_I^i \in \{0,1\}^i$
  denote the size-$i$ vector that has 1-entries at all indices in~$I$
  and 0-entries elsewhere.

  \begin{figure}[t]
    \centering
    \begin{subfigure}{.25\linewidth}
      \centering
      \begin{tikzpicture}[scale=0.5]
        \draw[help lines] (0,0) grid (4,5);
        % original vectors
        \node at (0.5,4.5) {1};
        \node at (1.5,4.5) {1};
        \node at (2.5,3.5) {1};
        \node at (3.5,3.5) {1};
        \node at (0.5,2.5) {1};
        \node at (3.5,2.5) {1};
        \node at (0.5,1.5) {1};
        \node at (2.5,1.5) {1};
      \end{tikzpicture}
      \caption{Original matrix.}
      \label{fig:subS}
    \end{subfigure}
    \begin{subfigure}{.33\linewidth}
      \centering
      \begin{tikzpicture}[scale=0.5]
        \draw[help lines] (0,0) grid (8,6);
        % original vectors
        \draw[thick] (0,1) rectangle (4,6);
        \node at (0.5,5.5) {1};
        \node at (1.5,5.5) {1};
        \node at (2.5,4.5) {1};
        \node at (3.5,4.5) {1};
        \node at (0.5,3.5) {1};
        \node at (3.5,3.5) {1};
        \node at (0.5,2.5) {1};
        \node at (2.5,2.5) {1};
        % b-block
        \foreach \x in {4,5,6} {
          \node at (\x+0.5,5.5) {1};
        }
        \node at (7.5,0.5) {1};
      \end{tikzpicture}
      \caption{Case 1}
      \label{fig:sub_i}
    \end{subfigure}\\[1em]
    \begin{subfigure}{.75\linewidth}
      \centering
      \begin{tikzpicture}[scale=0.5]
        \draw[help lines] (0,0) grid (16,5);
        % original vectors
        \draw[thick] (0,0) rectangle (4,5);
        \node at (0.5,4.5) {1};
        \node at (1.5,4.5) {1};
        \node at (2.5,3.5) {1};
        \node at (3.5,3.5) {1};
        \node at (0.5,2.5) {1};
        \node at (3.5,2.5) {1};
        \node at (0.5,1.5) {1};
        \node at (2.5,1.5) {1};
        % a-stairs
        \foreach \y in {1,...,4} {
          \node at (2.5 + 2*\y, 5.5-\y) {1};
          \node at (3.5 + 2*\y, 5.5-\y) {1};
        }
        \foreach \y in {1,...,4} {
          \node at (12.5, 5.5-\y) {1};
        }
        \foreach \x in {13,14,15} {
          \node at (\x+0.5,4.5) {1};
        }
      \end{tikzpicture}
      \caption{Case 2}
      \label{fig:sub_ii}
    \end{subfigure}
    \caption{Example of the construction for $\alpha=3$, $\beta=3$.}
    \label{fig:padding}
  \end{figure}
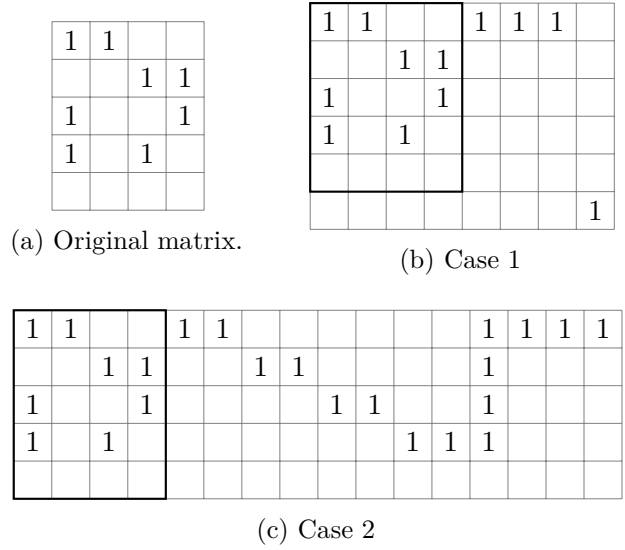

  We prove the theorem in two steps: First, the case~$\alpha=1$,
  $\beta=4+b$ for some~$b\ge 0$,
  and second, the case~$\alpha=2+a$, $\beta=4+2\lceil a/2\rceil+b$
  for some~$a,b\ge 0$.
  Note that these two cases together yield the statement of the theorem.

  \textbf{Case 1} ($\alpha=1$, $\beta=4+b$, $b\ge 0$).
  We define the instance $(S',k')$ as follows:
  The column vectors of the matrix~$S'\in\{0,1\}^{(n+1)\times
    (d+b+1)}$ are
  \[\col{s'}{j}:=\begin{cases}
    \mbf{1}_{D_j}^{n+1}, &j\in\{1,\ldots,d\},\\
    \mbf{1}_{\{1\}}^{n+1}, &j\in\{d+1,\ldots,d+b\},\\
    \mbf{1}_{\{n+1\}}^{n+1}, &j=d+b+1.
  \end{cases}\]
  We set~$k':=k+1$.
  An example of the constructed instance is shown in
  \autoref{fig:sub_i}.
  It is not hard to check that the rows of~$S'$ indeed fulfill the
  constraints on the Hamming distances:
  \begin{align*}
    h'&:=\min_{i\neq j\in [n+1]}\Ham(\row{s}{i},\row{s}{j})=\Ham(\row{s'}{n},\row{s'}{n+1})=1,\\
    H'&:=\max_{i\neq j\in [n+1]}\Ham(\row{s}{i},\row{s}{j})=\Ham(\row{s'}{1},\row{s'}{2})=\Ham(\row{s}{1},\row{s}{2})+b=4+b.
  \end{align*}
  As regards correctness, observe first that any solution
  contains the column index~$d+b+1$ because the row vectors~$\row{s'}{n}$ and
  $\row{s'}{n+1}$ only differ in this column.
  Since this column also distinguishes row~$\row{s}{n+1}$ from all
  other rows and no other pair of rows in~$S'$, it follows that
  $(S',k')$ is a yes-instance if and only if~$(\Scol[S']{[d+b]},k)$
  is a yes-instance. Due to \autoref{lem:padding}, this is the case if
  and only if~$(S,k)$ is a yes-instance.

  \textbf{Case 2} ($\alpha=2+a$, $\beta=4+2\lceil a/2\rceil+b$, $a,b\ge 0$).
  We define the instance $(S',k')$ as follows:
  Starting with~$S':=S$, we add~$\lceil a/2\rceil$ copies of the column
  vector~$\mbf{1}_{\{i\}}^n$ for each~$i\in[n-1]$ to~$S'$.
  Moreover, we add~$\lfloor a/2\rfloor$ copies of the column
  vector~$\mbf{1}_{[n-1]}^n$ to~$S'$.
  Finally, we add~$b$ copies of the column vector~$\mbf{1}_{\{1\}}^n$ to~$S'$
  and set~$k'=k$.
  \autoref{fig:sub_ii} shows an example of the construction.
  Indeed, we have the following Hamming distances:
  \begin{align*}
    \Ham(\row{s'}{n},\row{s'}{1})&=2+a+b,\\
    \Ham(\row{s'}{n},\row{s'}{j})&=2+a,\\
    \Ham(\row{s'}{1},\row{s'}{j})&=\Ham(\row{s}{1},\row{s}{j})+2\lceil
    a/2\rceil+b,\\
    \Ham(\row{s'}{j},\row{s'}{j'})&=\Ham(\row{s}{j},\row{s}{j'})+2\lceil
    a/2\rceil,
  \end{align*}
  for all~$j, j'\in\{2,\ldots,n-1\}$, $j \neq j'$.
  Thus, it holds $h'=2+a$ and $H'=4 + 2\lceil a/2\rceil+ b$.
  Since all row vectors in~$S$ are distinct and since we only added
  columns which distinguish exactly one row from all others, the correctness
  follows due to \autoref{lem:padding}.
\end{proof}

  \autoref{thm:abDV_NP-hard} yields the NP-completeness of
  \hhDV{\alpha}{\beta} for all~$\beta \ge \alpha+2$ (for
  even~$\alpha$) and $\beta\ge \alpha+3$ (for odd~$\alpha$),
  that is, for a given instance with fixed minimum pairwise row
  Hamming distance~$\alpha$,
  it is possible to increase the maximum pairwise row Hamming
  distance~$\beta$ arbitrarily without changing the answer to the instance.
  On the contrary, however, it seems impossible to construct an equivalent instance
  where only the minimum pairwise row Hamming distance is increased.
  Indeed, in the following, we show polynomial-time solvability
  for the case~$\alpha \ge 2\lfloor\beta/2\rfloor-1$.
  % Hence, unless P$=$NP, there is no way to bring the minimum
  % pairwise row Hamming distance arbitrarily close to the maximum
  % pairwise row Hamming distance in polynomial time.

  \subsection{Polynomial-Time Solvability for Homogeneous Data}\label{sec:binDVpoly}
  The polynomial-time algorithm for homogeneous is based on the observation that, for small differences
  between the values
  of minimum and maximum pairwise row Hamming distance, the input matrix
  is either highly structured or bounded in size by a constant
  depending only on the maximum pairwise row Hamming distance.
  This structure in turn guarantees that the instance is easily
  solvable.
  Before proving the theorem, we start with some basic
  results.

  First, we show that there is a linear-time preprocessing
  of a given input matrix such that the resulting matrix contains
  the null vector as a row and no two column vectors are identical.

  % \begin{observation}
  %   \label{obs:max-solution-size}
  %   Let~$I:=(S,k)$ be a \DV{} instance. If~$k \geq n-1$,
  %   then~$I$ is a yes-instance.
  % \end{observation}

  % \begin{proof}
  %   We have to show that~$n-1$ columns are always sufficient
  %   to distinguish~$n$ distinct rows. We prove this by a simple
  %   induction: For~$n=2$, this is clearly true since two distinct
  %   row vectors can always be distinguished by a single column.
  %   Now, let~$S$ be a matrix with~$n+1$ distinct rows and let~$S'$ be
  %   a submatrix containing all but one of the rows of~$S$,
  %   say~$x$.
  %   From the inductive hypothesis, it follows that the rows in~$S'$
  %   can be distinguished by at most~$n-1$ columns.
  %   In these columns, row~$x$ equals at most one other
  %   row~$y$ in~$S'$.
  %   Hence, adding another column in which~$x$ and~$y$ differ
  %   yields a solution of size at most~$n$.
  % \end{proof}
%
  \begin{lemma}
    \label{lem:preprocessing}
    For a given \DV{} instance~$I =(S\in\{0,1\}^{n\times d},k)$ one
    can compute in~$O(nd)$ time an equivalent \DV{} instance $I':=(S'\in\{0,1\}^{n\times d'},k)$
    such that~$S'$ contains the null vector~$\mathbf{0}\in\{0\}^{d'}$ as a row, the number~$d'$ of columns of~$S'$ is at most~$d$, and
    no two column vectors of~$S'$ are identical (implying~$d'\le 2^{n}$).
  \end{lemma}

  \begin{proof}
    From an instance~$I=(S,k)$, we compute~$S'$ as follows:
    First, in order to have the null vector~$\mathbf{0}$ as a row,
    we consider an arbitrary row vector, say~$\row{s}{1}$,
    and iterate over all columns~$j$.
    If~$s_{1j}=1$, then we exchange all~1's and~0's in column~$j$.
    Then, we sort the columns of~$S$ lexicographically in~$O(nd)$
    time (using radix sort). 
    We iterate over all columns again and check for any two successive
    column vectors whether they are identical and, if so, remove one of them.
    This ensures that all remaining column vectors are different,
    which implies that there are at most~$2^n$.
    Thus, in~$O(nd)$ time, we end up with a matrix~$S'$
    containing at most~$2^n$ columns, where~$\row{s'}{1}=\mbf{0}$.
    Clearly, reordering columns, removing identical columns, as well as
    exchanging 1's and 0's in a column does not change the
    answer to the original instance.
  \end{proof}

  We henceforth assume all
  input instances to be already preprocessed according to
  \autoref{lem:preprocessing}.
  In fact, we can extend \autoref{lem:preprocessing} by removing also
  inessential columns (recall \autoref{def:redundant}), that is, we can use the following data reduction rule.

  \begin{rrule}
    \label{rr:padding}
    Let $(S,k)$ be a \DV{} instance.
    If $S$ contains an inessential column, then delete this column
    from $S$.
  \end{rrule}
  \autoref{lem:padding} guarantees the correctness\footnote{A reduction rule is correct if it
    transforms yes-instances and only yes-instances into
    yes-instances.} of \autoref{rr:padding}. Exhaustive application of \autoref{rr:padding}
  can be done as follows: First, we determine in~$O(nd)$ time which
  columns fulfill Condition~(\ref{cond1}) of \autoref{def:redundant}.
  Recall that these are exactly the weight-1 and weight-$(n-1)$
  columns of which there can be at most~$\min\{n,d\}$ after the preprocessing
  according to \autoref{lem:preprocessing}.
  For each of these candidate columns~$j$, we check in~$O(nd)$ time
  whether Condition~(\ref{cond2}) also holds, that is, whether all row
  vectors are still distinct without column~$j$,
  by lexicographically sorting the rows of the matrix without
  column~$j$. The overall running time is thus in~$O(\min\{n,d\}\cdot nd)$.
  % As exhaustively applying \autoref{rr:padding} will not increase our
  % running times below,
  % we may assume that all instances are preprocessed also with respect to \autoref{rr:padding}.
  
  We now turn towards proving polynomial-time solvability of \hhDV{\alpha}{\beta} for~$\alpha\ge
  2\lfloor\beta/2\rfloor -1$.
  The proof uses some results from extremal
  combinatorics concerning certain set systems.
  We refer the reader to the book by \citet[Chapter~6]{jukna11} for an
  introduction into this topic.
  To start with, we introduce the necessary concepts and notation.
  Recall \autoref{def:weight}, where we defined the set~$W_{\row{s}{i}}$ of
  column indices where row~$i$ equals~1. In the following, for a given
  input matrix~$S$ and a given set of row indices~$I$, we will consider the
  \emph{column system} of~$I$, that is, the system containing the sets~$W_{\row{s}{i}}$ of
  column indices of all row vectors with indices in~$I$.
  
  \begin{definition}
    \label{def:colsys}
    For a matrix $S\in\{0,1\}^{n\times d}$ and a
    subset~$I\subseteq[n]$ of row indices,
    let~$\setsys{}(I) := \{W_{\row{s}{i}} \mid  i\in
    I\}$ denote the \emph{column system} of~$I$ containing the
    sets~$W_{\row{s}{i}}$ of column indices for all rows in~$I$.
    For~$\omega\in[d]$, let~$I_\omega:=\{i\in[n]\mid
    w(\row{s}{i})=\omega\}$ be the set of indices of the
    weight-$\omega$ rows and
    let~$\setsys{\omega}:=\setsys{}(I_\omega)$ be the column system
    of the weight-$\omega$ rows.
  \end{definition}
  \begin{figure}[t]
    \centering
    \begin{tikzpicture}[scale=.4]
        \draw[help lines] (0,0) grid (7,5);
        \foreach \x in {1,...,7} {
          \node at (\x-0.5,5.5) {\sffamily\scriptsize\x};
        }
        \node at (0.5,4.5) {1};
        \node at (0.5,3.5) {1};
        \node at (6.5,3.5) {1};
        \node at (1.5,1.5) {1};
        \node at (2.5,2.5) {1};
        \node at (4.5,1.5) {1};
        \node at (3.5,3.5) {1};
        \node at (3.5,1.5) {1};
        \node at (4.5,4.5) {1};
        \node at (5.5,0.5) {1};
        \node at (6.5,0.5) {1};

        \node[right] at (8,4.5) {$\setsys{1}=\{\{3\}\}$};
        \node[right] at (8,2.5) {$\setsys{2}=\{\{1,5\},\{6,7\}\}$};
        \node[right] at (8,0.5) {$\setsys{3}=\{\{1,4,7\},\{2,4,5\}\}$};
      \end{tikzpicture}
    \caption{An example of a binary matrix (left) containing rows of
      weight one, two, and three. The corresponding row systems are
      written on the right.}
    \label{fig:row-system-example}
  \end{figure}
  \autoref{fig:row-system-example} illustrates \autoref{def:colsys}.
  Note that in order to distinguish all rows of weight~$\omega$ from
  each other, we only have to consider those columns which
  appear in some of the sets contained in the column system~$\setsys{\omega}$
  since the weight-$\omega$ rows only differ in these columns.
  Thus, in order to find subsolutions for the weight-$\omega$ rows,
  the structure of~$\setsys{\omega}$, especially the pairwise
  intersections of the contained sets, will be very important for us.
  Therefore, we make use of two general combinatorial
  concepts of set systems~\cite{jukna11}, the first of which defines a system
  of sets that pairwise intersect in the same number of elements,
  whereas the second concept describes the even stronger condition that
  all pairwise intersections contain the same elements.

  \begin{definition}[Weak $\Delta$-system]
    A family $\mathcal{F} = \{S_1,\ldots,S_m\}$ of~$m$ different sets is called a
    \emph{weak}~$\Delta$-system if there is some~$\lambda\in\N$ such
    that~$|S_i \cap S_j| = \lambda$ for all~$i \neq j\in[m]$.
  \end{definition}

  \begin{definition}[Strong $\Delta$-system]
    A \emph{strong} $\Delta$-system (or \emph{sunflower}) is a
    weak~$\Delta$-system~$\{S_1,\ldots,S_m\}$ such that~$S_i\cap S_j = C$
    for all~$i\neq j\in[m]$ and some set~$C$ called the \emph{core}.
    The sets~$\pet{S}{i}:=S_i \setminus C$ are called \emph{petals}.
  \end{definition}

  As a first case, the following lemma illustrates the merit of the
  above definitions showing that any \DV{} instance can
  easily be solved if the underlying column system of all
  non-zero-weight rows forms a sunflower.
  
  \begin{lemma}
    \label{lem:sunflower-solution}
    Let~$I:=(S\in\{0,1\}^{n\times d},k)$ be a \DV{} instance such
    that~$\setsys{}:=\bigcup_{\omega \ge
      1}\setsys{\omega}$
    forms a sunflower (note that $\setsys{0}=\emptyset\not\in\setsys{}$).
    Then, $I$ is a yes-instance if and only if~$k\ge|\setsys{}|$.
    Moreover, any solution intersects at least all but one of the
    petals of~$\setsys{}$.
  \end{lemma}

  \begin{proof}
    Recall that we assume the instance~$I$ to be already preprocessed according to \autoref{lem:preprocessing}. Hence, we can assume without loss of generality, that in~$S$, $\row{s}{n}=\mbf{0}$ and no two column vectors
    are equal; assume further that~$\setsys{}=\{W_{\row{s}{1}},\ldots,W_{\row{s}{n-1}}\}$ is
    a sunflower with core~$C$. An example is depicted in
    \autoref{fig:sunflower-instance}.

    Recall that any solution~$K$ fulfills~$K\cap D_{ij}\neq\emptyset$
    for all~$i\neq j\in[n]$, where~$D_{ij}$ is the set of column indices in which the row vectors~$\row{s}{i}$ and~$\row{s}{j}$ differ. Assume towards a contradiction that~$K\subseteq[d]$
    with~$|K|< n-1$ is a solution. If~$K\cap C=\emptyset$, then~$K$
    only intersects the petals. Since the petals are pairwise
    disjoint, it follows that there exists an~$i\in[n-1]$ such
    that~$K\cap W_i= K\cap D_{in}=\emptyset$, which shows that~$K$
    cannot be solution. If~$K\cap C \neq \emptyset$, then~$K$
    intersects at most~$n-3$ of the~$n-1$ petals in~$\setsys{}$.
    Hence, there exist~$i, j\in[n-1]$ with~$i\neq j$ such
    that~$K\cap (\pet{W}{i}\cup\pet{W}{j})=K\cap D_{ij}=\emptyset$.
    Hence, $K$ cannot be a solution.
    It remains to show that there is always a solution of
    size~$|\setsys{}|=n-1$. To this end, let~$K$ contain an arbitrary
    element from each non-empty petal and, if there is an empty petal,
    also an arbitrary element from the core~$C$. Clearly, $K$ is a
    solution of size~$n-1$.
  \end{proof}

  \begin{figure}[t]
    \centering
    \begin{tikzpicture}[scale=.5]
        \draw[help lines] (0,0) grid (10,7);
        \foreach \x in {1,...,10} {
          \node at (\x-0.5,7.5) {\sffamily\scriptsize\x};
        }
        \foreach \y in {0,...,5}{
          \node at (0.5,6.5-\y) {1};
          \node at (1.5,6.5-\y) {1};
        }
        \node at (2.5,4.5) {1};
        \node at (3.5,6.5) {1};
        \node at (6.5,5.5) {1};
        \node at (5.5,3.5) {1};
        \node at (7.5,3.5) {1};
        \node at (4.5,2.5) {1};
        \node at (8.5,6.5) {1};
        \node at (9.5,4.5) {1};
        \draw[thick] (1,0) rectangle (7,8);
      \end{tikzpicture}
    \caption{Example of a matrix where the set system~$\mathcal{D}$
      forms a sunflower with core~$C=\{1,2\}$ consisting of the first two
      columns. The six petals from top to bottom are~$\{4,9\}$, $\{7\}$, $\{3,10\}$,
      $\{6,8\}$, $\{5\}$ and~$\emptyset$.
      Framed by thick lines is a set of columns that distinguish all rows.}
    \label{fig:sunflower-instance}
  \end{figure}
  According to \autoref{lem:sunflower-solution}, identifying
  sunflower structures in a given input instance significantly
  simplifies our problem since they have easy solutions.
  To this end, the following result by \citet{deza74} will serve as
  an important tool since it describes conditions under which a
  weak~$\Delta$-system actually becomes a strong one, that is, a sunflower (see also \citet[Chapter~6, Theorem~6.2]{jukna11}).
  
  \begin{lemma}[{\citet[Theorem 2]{deza74}}]
    \label{lem:deza}
    Let~$\mathcal{F}$ be an $s$-uniform weak~$\Delta$-system, that is,
    each set contains~$s$ elements.
    If~$|\mathcal{F}| \geq s^2 - s + 2$, then~$\mathcal{F}$ is a sunflower.
  \end{lemma}
  The basic scheme for proving polynomial-time solvability of
  \hhDV{\alpha}{\beta} for~$\alpha\le\beta <2\lceil\alpha/2\rceil+2$
  is the following: The bounds on the minimum and maximum
  pairwise row Hamming distances imply that the column
  systems~$\setsys{x}$ for~$x = \alpha,\ldots,\beta$ form $x$-uniform
  weak~$\Delta$-systems. Using \autoref{lem:deza}, we then conclude
  that either the size of the instance is bounded by a constant
  depending on~$\beta$ only, or that the~$\setsys{x}$
  form sunflowers, which we can handle according to
  \autoref{lem:sunflower-solution}.

  As a final prerequisite, we prove the following easy but helpful
  lemma, concerning the intersection of sets with sunflowers.
  
  \begin{lemma}
    \label{lem:sunflower-intersect}
    Let~$\lambda \in \N$, let $\mathcal{F}$ be a sunflower with core~$C$ and let~$X$
    be a set such that~$|X \cap S| \ge \lambda$ for
    all~$S\in\mathcal{F}$.
    If~$|\mathcal{F}| > |X|$, then~$\lambda\le |C|$ and~$|X\cap C|\ge \lambda$.
  \end{lemma}

  \begin{proof}
    Assume towards a contradiction that~$|X \cap C|<\lambda$.
    Then $X$ would intersect each of the~$|\mathcal{F}|>|X|$ pairwise
    disjoint petals of~$\mathcal{F}$, which is not possible.
  \end{proof}
  We are now ready to prove the following theorem.

  \begin{theorem}
    \label{thm:hhDV-poly}
    \hhDV{\alpha}{\beta} is solvable
    \begin{enumerate}[1.)]
      %\item\label{thm:hhDV-poly_i} $O(nd^2)$ time if~$\beta = \alpha$,
      \item\label{thm:hhDV-poly_ii} in $O(\min\{n,d\}\cdot nd)$ time if~$\beta\le\alpha+1$, and 
      \item\label{thm:hhDV-poly_iii} in $O(n^3d)$ time if 
        $\alpha$ is odd and~$\beta=\alpha+2$.
    \end{enumerate}
  \end{theorem}
  
  We prove both statements of \autoref{thm:hhDV-poly} separately. As mentioned, the basic structures of both proofs are similar: We first partition the column system into uniform weak $\Delta$-systems. Then, we consider each of the cases of which of the systems are sunflowers or of bounded size. Then, we leverage the preprocessing (\autoref{lem:preprocessing} and \autoref{rr:padding}) and our knowledge of solutions for sunflowers (\autoref{lem:sunflower-solution}) to show that only a small number of possible solutions are left. That is, the instances are essentially solved by the preprocessing routines and we can try out all remaining solutions to solve the instances in polynomial time. Showing that the remaining possible solutions are few seems more difficult for Statement~(\ref{thm:hhDV-poly_iii}.); thus, the proof of Statement~(\ref{thm:hhDV-poly_ii}.) can be seen as a ``warm-up''.

  \begin{proof}[Proof \normalfont(\autoref{thm:hhDV-poly}, Statement~(\ref{thm:hhDV-poly_ii}.)).]
  In the following, let~$I := (S\in\{0,1\}^{n\times d},k)$ be an instance of
  \hhDV{\alpha}{\beta} for~$\alpha\le\beta\le\alpha+1$.
  Recall that we assume~$I$ to be already preprocessed according to
  \autoref{lem:preprocessing} and \autoref{rr:padding}
  in~$O(\min\{n,d\}\cdot nd)$ time, that is,
  $S$~contains the null row vector, say~$\row{s}{n}=\mbf{0}$,
  no two column vectors are equal, which implies~$d\le 2^n$, and
  there are no inessential columns.
  We write~$W_i := W_{\row{s}{i}}=\{j\in[d]\mid s_{ij}=1\}$ for the
  set of column indices~$j$ where row vector~$\row{s}{i}$ equals~1,
  and we define~$W_{ij}:=W_i \cap W_j$.
  For~$\omega\in[d]$, let~$I_\omega:=\{i\in[n]\mid w(\row{s}{i})=\omega\}$ denote
  the set of indices of the weight-$\omega$ rows and
  let~$n_\omega:=|I_\omega|$.
  For ease of presentation, we sometimes identify columns or rows and their corresponding indices. 

    For all~$i\in[n-1]$, we have $$\Ham(\row{s}{i},\row{s}{n})=\Ham(\row{s}{i},\mbf{0})=
    w(\row{s}{i})\in\{\alpha,\alpha+1\}\text{.}$$
    Since also~$\Ham(\row{s}{i},\row{s}{j})=w(\row{s}{i}) + w(\row{s}{j}) -
    2|W_{ij}|\in\{\alpha,\alpha+1\}$ for
    all~$i\neq j\in[n-1]$, the following properties can be derived:
    \begin{align}
      \label{eqn:iiaa}\forall i,j\in I_\alpha,\;i\neq j : |W_{ij}| &= \lfloor \alpha/2 \rfloor,\\
      \label{eqn:iia1a1}\forall i,j\in I_{\alpha+1},\;i\neq j : |W_{ij}| &= \lceil (\alpha+1)/2 \rceil\text{, and}\\
      \label{eqn:iiaa1}\forall i\in I_\alpha,\;j\in I_{\alpha+1} : |W_{ij}| &= \lfloor (\alpha+1)/2 \rfloor.
    \end{align}
    For example, let us prove Property~(\ref{eqn:iiaa}). If $i, j \in I_\alpha$, then $2\alpha - 2|W_{ij}| \in \{\alpha, \alpha + 1\}$. If $\alpha$ is even, then, since $|W_{ij}|$ is an integer, $2\alpha - 2|W_{ij}| = \alpha$. Thus, $|W_{ij}| = \alpha/2 = \lfloor \alpha/2 \rfloor$. If $\alpha$ is odd, then $2\alpha - 2|W_{ij}| = \alpha + 1$ and, hence, $|W_{ij}| = (\alpha - 1)/2 = \lfloor \alpha/2 \rfloor$. This proves Property~(\ref{eqn:iiaa}). The proofs for the remaining properties are analogous.

    Property~(\ref{eqn:iiaa}) implies
    that~$\setsys{\alpha}:=\{W_i\mid i\in I_\alpha\}$ is an~$\alpha$-uniform weak~$\Delta$-system
    and Property~(\ref{eqn:iia1a1}) implies that $\setsys{\alpha+1}:=\{W_i\mid
    i\in I_{\alpha+1}\}$ is an~$(\alpha+1)$-uniform weak~$\Delta$-system.
    Let~$c:=(\alpha+1)^2-(\alpha+1)+2$.
    We can assume that~$\max\{n_\alpha,n_{\alpha+1}\} \ge c$
    because otherwise~$n\le 2c$ is of constant size, and thus also
    $d\le 2^n$ is of constant size (recall that we assume~$I$
    to be preprocessed according to \autoref{lem:preprocessing}),
    which implies that~$I$ is constant-time solvable.

    \medskip
    
    First, consider the case that~$n_{\alpha+1} \ge c$.
    Then, by \autoref{lem:deza}, it follows that~$\setsys{\alpha+1}$ is a
    sunflower with a core~$C$ of size~$\lceil (\alpha+1)/2 \rceil$ and
    petals~$\pet{W}{i}$, $i\in I_{\alpha+1}$, of size~$\alpha+1-|C|\ge 1$.
    For each~$i\in I_{\alpha+1}$ and each~$j\in I_\alpha$, it follows by
    Property~(\ref{eqn:iiaa1}) and \autoref{lem:sunflower-intersect}
    that~$W_{ij}\subseteq C$, that is~$\pet{W}{i}\cap W_j=\emptyset$.
    Hence, for~$x\in\pet{W}{i}$, the column vector~$\col{s}{x}$
    contains exactly one 1 (namely in the $i$-th row), that is,
    $\col{s}{x}=\mbf{1}^n_{\{i\}}$. Thus, column~$x$ exactly
    distinguishes row~$i$ from all other rows.
    For~$\alpha \ge 2$, each pair of rows differs in at least two
    columns. Thus, all rows in~$\Scol{[d]\setminus\{x\}}$ are still
    distinct and column~$x$ is in fact inessential, which
    yields a contradiction.
    Hence, we can assume that~$\alpha=1$.
    Since, then, for all~$i\in I_1$ and~$j\in I_2$ we have~$W_{ij}=W_i=C$,
    it follows that~$n_1 = 1$. See \autoref{sub:iia} for an
    illustrating example.
    The only possible solution is thus~$K=\bigcup_{i\in[n-1]}W_i =
    [d]$.
    
    If $n_{\alpha+1} < c$, then~$n_\alpha \ge c$ holds
    and \autoref{lem:deza} implies that~$\setsys{\alpha}$
    is a sunflower with a core~$C$ of size~$|C|=\lfloor \alpha/2 \rfloor$.
    If~$(\alpha+1)$ is even, then we have~$\lfloor (\alpha+1)/2
    \rfloor > |C|$, and thus, by Property~(\ref{eqn:iiaa1}) and
    \autoref{lem:sunflower-intersect}, it follows~$n_{\alpha+1}=0$
    (see \autoref{sub:iib}).
    Now, by \autoref{lem:sunflower-solution},~$I$ is a yes-instance
    if and only if~$k\ge n_\alpha$.
    If~$\alpha$ is even, then~$|C|=\lfloor(\alpha+1)/2\rfloor$,
    and thus, Property~(\ref{eqn:iiaa1}) and
    \autoref{lem:sunflower-intersect} imply that~$W_{ij}=C$ for
    all~$i\in I_\alpha$, $j\in I_{\alpha+1}$.
    Note that~$\col{s}{x}=\mbf{1}^n_{[n-1]}$ for all~$x\in C$, that
    is, column~$x$ exactly distinguishes row~$n$ from all others.
    Since~$\alpha$ is even, hence~$\alpha \ge 2$, it follows that
    column~$x$ is inessential,
    which again yields a contradiction.
    \begin{figure}[t]
    \centering
    \begin{subfigure}{.4\linewidth}
      \centering
      \begin{tikzpicture}[scale=0.5]
        \draw[help lines] (0,0) grid (5,6);
        \node at (0.5,5.5) {1};
        \node at (0.5,4.5) {1};
        \node at (0.5,3.5) {1};
        \node at (0.5,2.5) {1};
        \node at (0.5,1.5) {1};
        
        \foreach \x in {1,...,4} {
          \node at (\x+0.5, 6.5-\x) {1};
        }
      \end{tikzpicture}
      \caption{}
      \label{sub:iia}
    \end{subfigure}
    \begin{subfigure}{.4\linewidth}
      \centering
      \begin{tikzpicture}[scale=0.5]
        \draw[help lines] (0,0) grid (5,6);
        \foreach \x in {1,...,5} {
          \node at (\x-0.5, 6.5-\x) {1};
        }
      \end{tikzpicture}
      \caption{}
      \label{sub:iib}
    \end{subfigure}
    \caption{Examples of two possible instances for the
      case~$\alpha=1$, $\beta=2$.}
    \label{fig:ii-example}
  \end{figure}
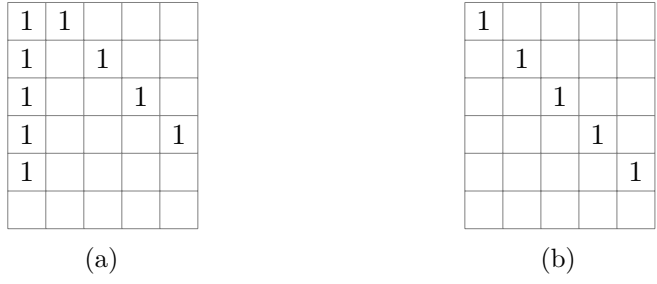
\end{proof}
Next, we show that \hhDV{\alpha}{\beta} is solvable in $O(n^3d)$ time if $\alpha$ is odd and~$\beta=\alpha+2$. We use the same notation as in the proof of Statement~(\ref{thm:hhDV-poly_ii}.).
\begin{proof}[Proof \normalfont({\autoref{thm:hhDV-poly}, Statement~(\ref{thm:hhDV-poly_iii}.))}.]
    \setcounter{equation}{0} 
  %  \textbf{\ref{thm:hhDV-poly_iii}.)} $\boldsymbol{\beta=\alpha+2}$ ($\alpha$ being odd):
    Since~$\Ham(\row{s}{i},\row{s}{j}) =
    w(\row{s}{i})+w(\row{s}{j})-2|W_{ij}|\in\{\alpha,\alpha+1,\alpha+2\}$
    holds for all~$i\neq j \in [n]$,
    it follows that~$\Ham(\row{s}{i},\row{s}{n})=w(\row{s}{i})\in\{\alpha,\alpha+1,\alpha+2\}$
    holds for all~$i\in[n-1]$.
    By plugging in the respective values for~$w(\row{s}{i})$ and~$w(\row{s}{j})$ in the above formula for~$\Ham(\row{s}{i},\row{s}{j})$, the following properties can be derived (in an analogous way as for Properties~(\ref{eqn:iiaa}) to~(\ref{eqn:iiaa1}) in the proof of Statement~(\ref{thm:hhDV-poly_ii})):
    \begin{align}
      \label{eqn:iiiaa}\forall i,j\in I_\alpha,\; i\neq j: |W_{ij}| &= \lfloor\alpha/2\rfloor,\\
      \label{eqn:iiiaa1}\forall i\in I_\alpha,\;j\in I_{\alpha+1} : |W_{ij}| &\in \{\lfloor\alpha/2\rfloor,\lceil\alpha/2 \rceil\},\\
      \label{eqn:iiiaa2}\forall i\in I_\alpha,\;j\in I_{\alpha+2} : |W_{ij}| &= \lceil\alpha/2\rceil,\\
      \label{eqn:iiia1a1}\forall i,j\in I_{\alpha+1},\; i\neq j : |W_{ij}| &= \lceil\alpha/2\rceil,\\
      \label{eqn:iiia1a2}\forall i\in I_{\alpha+1},\;j\in I_{\alpha+2} : |W_{ij}| &\in \{\lceil\alpha/2\rceil,\lceil\alpha/2 \rceil+1\},\\
      \label{eqn:iiia2a2}\forall i,j\in I_{\alpha+2},\; i\neq j : |W_{ij}| &= \lceil\alpha/2\rceil+1.
    \end{align}
    Properties (\ref{eqn:iiiaa}), (\ref{eqn:iiia1a1}), and~(\ref{eqn:iiia2a2})
    imply that~$\setsys{\alpha}$, $\setsys{\alpha+1}$,
    and~$\setsys{\alpha+2}$ are $\alpha$-, $(\alpha+1)$-, and
    $(\alpha+2)$-uniform weak~$\Delta$-systems, respectively.

    In the following, we denote by~$U_\omega:=\bigcup_{i\in I_\omega}W_i$ the index set of
    the columns where at least one weight-$\omega$ row vector equals~1. Let~$c:=(\alpha+2)^2-(\alpha+2)+2$.
    For each~$x\in\{\alpha,\alpha+1,\alpha+2\}$, we either have~$n_x <
    c$ or~$n_x \ge c$.
    Overall, this gives eight possible cases, each of which we now
    show how to solve:

%    \medskip
    
    \begin{inparaenum}[\bf {Case} I]%

      \item\label{case000} ($n_\alpha < c$, $n_{\alpha+1}<c$, $n_{\alpha+2}<c$).
        In this case, the number of rows
        in~$S$ is upper-bounded by a constant depending on~$\alpha$, and thus,
        $I$ is of overall constant size.

      \item\label{case100} ($n_\alpha \ge c$, $n_{\alpha+1}<c$, $n_{\alpha+2}<c$).
        Due to \autoref{lem:deza}, family~$\setsys{\alpha}$ forms a sunflower.
        Let~$C$ with~$|C|=\lfloor\alpha/2\rfloor$ be the
        core of~$\setsys{\alpha}$.
        For~$\alpha=1$, clearly,  any solution~$K$ contains all column
        indices from~$U_1$ in order to distinguish the weight-1 rows from
        the null vector. Since~$|U_2\cup U_3|\le 2n_2+3n_3$ is upper-bounded
        by a constant, the number of possible subsets~$K'\subseteq
        U_2\cup U_3$ is also upper-bounded by a constant.
        Thus, we only have to check a constant number of choices~$K= U_1\cup K'$.
        For~$\alpha\ge 3$, the size of a petal~$\pet{W}{i}$, $i\in
        I_\alpha$, is~$|\pet{W}{i}|=|W_i|-|C| = \alpha -
        \lfloor\alpha/2\rfloor= \lceil\alpha/2\rceil\ge 2$.
        Since the petals are pairwise disjoint, it follows that,
        for each petal~$\pet{W}{i}$, there exists a~$j\in I_{\alpha+1}\cup
        I_{\alpha+2}$ such that~$\pet{W}{i}\cap W_j\neq
        \emptyset$: Otherwise, the column vectors corresponding to the
        indices in a petal~$\pet{W}{i}$ are all equal
        to~$\mbf{1}^n_{\{i\}}$, that is, at least one of them is
        inessential, which is a contradiction.
        Since~$|U_{\alpha+1}\cup U_{\alpha+2}|$ is upper-bounded by a constant
        depending on~$\alpha$, also the number~$n_{\alpha}$ of
        petals in~$\setsys{\alpha}$ is upper-bounded by a constant, which yields
        an overall constant size of~$I$.

      \item\label{case010} ($n_\alpha < c$, $n_{\alpha+1}\ge c$, $n_{\alpha+2}<c$).
        By Property~(\ref{eqn:iiia1a1}) and \autoref{lem:deza}, family~$\setsys{\alpha+1}$ forms a sunflower with
        a core~$C$ of size~$|C|=\lceil\alpha/2\rceil$.
        The size of each petal~$\pet{W}{i}$, $i\in
        I_{\alpha+1}$, is thus~$|\pet{W}{i}|= \lceil\alpha/2\rceil$.
        Hence, for~$\alpha\ge 3$, the same arguments as in
        Case~(\ref{case100}) hold.
        For~$\alpha=1$, any solution~$K$
        can be written as~$K=C' \cup U_1\cup K_2 \cup K_3$,
        where~$C'\subseteq C$, $K_2\subseteq U_2\setminus C$
        and~$K_3\subseteq U_3$.
        Note that~$|C|$ and~$|U_3|$ are upper-bounded by a constant.
        Hence, the number of different subsets~$C'$ and~$K_3$ is also a
        constant.
        Since~$|\pet{W}{i}|=1$ holds for all~$i\in I_2$, we have~$|U_2
        \setminus C| = n_2$.
        From \autoref{lem:sunflower-solution}, it follows
        that~$|K_2|\ge n_2-1$.
        The overall number of possible choices for~$K_2$, and thus
        for~$K$, is in~$O(n)$.   

      \item\label{case001} ($n_\alpha < c$, $n_{\alpha+1}<c$, $n_{\alpha+2}\ge c$).
        By \autoref{lem:deza}, family~$\setsys{\alpha+2}$ forms a sunflower with
        core~$C$ of size~$|C|=\lceil\alpha/2\rceil+1$.
        The size of each petal~$\pet{W}{i}$, $i\in
        I_{\alpha+2}$, is thus~$|\pet{W}{i}|= \lceil\alpha/2\rceil$.
        Hence, for~$\alpha\ge 3$, the same arguments as in
        Case~(\ref{case100}) hold. For~$\alpha=1$, any solution~$K$
        can be written as~$K=C' \cup U_1\cup K_2 \cup K_3$,
        where~$C'\subseteq C$, $K_2\subseteq U_2$ and~$K_3\subseteq
        U_3\setminus C$. Note that~$|C|$ and~$|U_2|$ are upper-bounded by a constant.
        Hence, the number of different subsets~$C'$ and~$K_2$ is also a
        constant. \autoref{lem:sunflower-solution}
        implies that~$|K_3|\ge n_3 -1$. Since~$|U_3 \setminus C|=n_3$,
        this yields an overall number of~$O(n)$ possible choices for~$K$.

      \item\label{case110} ($n_\alpha \ge c$, $n_{\alpha+1}\ge c$, $n_{\alpha+2}< c$).
        Due to \autoref{lem:deza}, family~$\setsys{\alpha}$ forms a sunflower
        with a core~$C$ of
        size~$|C|=\lfloor\alpha/2\rfloor$ and~$\setsys{\alpha+1}$
        forms a sunflower with core~$C'$ of size~$|C'|=\lceil\alpha/2\rceil$.
        First, note that Property~(\ref{eqn:iiiaa2})
        implies~$|W_{ij}|=\lceil\alpha/2\rceil > |C|$
        for all~$i\in I_\alpha$, $j\in I_{\alpha+2}$,
        which is not possible due to \autoref{lem:sunflower-intersect}.
        Thus, it follows~$n_{\alpha+2}=0$.
        Moreover, since Property~(\ref{eqn:iiiaa1})
        implies~$|W_{ij}|\ge|C|$ for all~$i\in I_\alpha$ and $j\in
        I_{\alpha+1}$, \autoref{lem:sunflower-intersect} yields
        $C\subseteq W_{ij}$, and thus~$C\subset C'$.
        Hence, all column vectors in~$C$ equal~$\mbf{1}^n_{[n-1]}$,
        which yields a contradiction for~$\alpha \ge 3$ because
        the columns in~$C$ are then inessential.
        For~$\alpha=1$, any solution~$K$ can be written as
        $K = C''\cup U_1\cup K_2$, where~$C''\subseteq C'$
        and~$K_2\subseteq U_2$.
        By \autoref{lem:sunflower-solution}, we know that~$|K_2|\ge
        n_2-1$. Since~$|C'|=1$ and~$|U_2\setminus C'|=n_2$, there are $O(n)$ possible choices for~$K$.

      \item\label{case101} ($n_\alpha \ge c$, $n_{\alpha+1}<c$, $n_{\alpha+2}\ge c$).
        This case is not possible since we showed in
        Case~\ref{case110} that~$n_\alpha\ge c$ implies~$n_{\alpha+2}=0$.
      
      \item\label{case111} ($n_\alpha \ge c$, $n_{\alpha+1}\ge c$, $n_{\alpha+2}\ge c$).
        This case is also not possible, see Case~\ref{case110}.

      \item\label{case011} ($n_\alpha < c$, $n_{\alpha+1}\ge c$, $n_{\alpha+2}\ge c$).
        From \autoref{lem:deza} and Properties~(\ref{eqn:iiia1a1}) and~(\ref{eqn:iiia2a2}), respectively, it follows that~$\setsys{\alpha+1}$
        forms a sunflower with a core~$C$ of
        size~$|C|=\lceil\alpha/2\rceil$ and~$\setsys{\alpha+2}$
        forms a sunflower with core~$C'$ of
        size~$|C'|=\lceil\alpha/2\rceil+1$.
        Moreover, as in Case~\ref{case110}), Property~(\ref{eqn:iiia1a2}) and
        \autoref{lem:sunflower-intersect} imply~$C\subset C'$.

        If~$\alpha=1$, then any solution can be written as~$K=U_1\cup
        C''\cup K_2\cup K_3$, where~$C''\subseteq C'$, $K_2\subseteq
        U_2\setminus C$ and~$K_3\subseteq U_3\setminus C'$.
        Since~$|C'|=\lceil\alpha/2\rceil+1$, $|U_2\setminus C|=n_2$,
        $|U_3\setminus C'|=n_3$, and, by
        \autoref{lem:sunflower-solution}, $|K_2|\ge n_2-1$
        and~$|K_3|\ge n_3-1$, it follows that there are $O(n^2)$ possible
        choices for~$K$.

        \begin{figure}[t]
          \centering
           \begin{tikzpicture}[scale=0.5]
              \draw[fill=gray!20!white,opacity=.7] (2,0) rectangle (3,10);
              \draw[fill=gray!20!white,opacity=.7] (3,0) rectangle (4,10);
              \draw[fill=gray!20!white,opacity=.7] (5,0) rectangle (6,10);
              \draw[fill=gray!20!white,opacity=.7] (8,0) rectangle (9,10);
              \draw[fill=gray!20!white,opacity=.7] (9,0) rectangle (10,10);
              \draw[help lines] (0,0) grid (11,10);
              \node at (0.5,9.5) {1};
              \node at (0.5,8.5) {1};
              \node at (0.5,7.5) {1};
              \node at (0.5,6.5) {1};
              \node at (0.5,5.5) {1};
              \node at (0.5,4.5) {1};
              \node at (0.5,3.5) {1};
              \node at (0.5,1.5) {1};

              \node at (1.5,9.5) {1};
              \node at (1.5,8.5) {1};
              \node at (1.5,7.5) {1};
              \node at (1.5,6.5) {1};
              \node at (1.5,5.5) {1};
              \node at (1.5,4.5) {1};
              \node at (1.5,3.5) {1};
              \node at (1.5,2.5) {1};

              \node at (2.5,9.5) {1};
              \node at (2.5,8.5) {1};
              \node at (2.5,7.5) {1};
              \node at (2.5,2.5) {1};
              \node at (2.5,1.5) {1};

              \node at (3.5,9.5) {1};
              \node at (4.5,9.5) {1};
              \node at (5.5,8.5) {1};
              \node at (6.5,8.5) {1};
              \node at (7.5,7.5) {1};
              \node at (8.5,7.5) {1};

              \node at (3.5,6.5) {1};
              \node at (6.5,6.5) {1};
              \node at (4.5,5.5) {1};
              \node at (8.5,5.5) {1};
              \node at (7.5,4.5) {1};
              \node at (9.5,4.5) {1};
              \node at (5.5,3.5) {1};
              \node at (10.5,3.5) {1};
              
              \node at (9.5,2.5) {1};
              \node at (10.5,1.5) {1};

              \draw[thick] (3,1) rectangle (11,10);
              \draw[decorate,decoration={brace, amplitude=5pt, mirror}] (0,0)
              -- (3,0) node [midway, yshift=-10pt] {$C'$};
              \draw[decorate,decoration={brace, amplitude=5pt}] (0,10)
              -- (2,10) node [midway, yshift=10pt] {$C$};
              \node[above] at (2.5,10) {$z$};
              \draw[decorate,decoration={brace, amplitude=5pt, mirror}] (0,10)
              -- (0,7) node [midway, xshift=-10pt] {$I_5$};
              \draw[decorate,decoration={brace, amplitude=5pt,mirror}] (0,7)
              -- (0,3) node [midway, xshift=-10pt] {$I_4$};
              \draw[decorate,decoration={brace, amplitude=5pt,mirror}] (0,3)
              -- (0,1) node [midway, xshift=-10pt] {$I_3$};

              \node[draw, circle, inner sep=3pt] (51) at (14,8.0) {};
              \node[draw, circle, inner sep=3pt] (52) at (15.5,8.0) {};
              \node[draw, circle, inner sep=3pt] (53) at (17,8.0) {};

              \draw[decorate,decoration={brace, amplitude=5pt}] (13.5,8.5)
              -- (17.5,8.5) node [midway, yshift=10pt] {$I_5$};

              \node[draw, circle, inner sep=3pt] (31) at (19,8.0) {};
              \node[draw, circle, inner sep=3pt] (32) at (20.5,8.0) {};

              \draw[decorate,decoration={brace, amplitude=5pt}] (18.5,8.5)
              -- (21,8.5) node [midway, yshift=10pt] {$I_3$};
              
              \node[draw, circle, inner sep=3pt] (41) at (15,4.0) {};
              \node[draw, circle, inner sep=3pt] (42) at (16.5,4.0) {};
              \node[draw, circle, inner sep=3pt] (43) at (18,4.0) {};
              \node[draw, circle, inner sep=3pt] (44) at (19.5,4.0) {};

              \draw[decorate,decoration={brace, amplitude=5pt, mirror}] (14.5,3.5)
              -- (20,3.5) node [midway, yshift=-11pt] {$I_4$};
              
              \draw[line width=2pt] (51) -- (41);
              \draw (51) -- (42);
              \draw (52) -- (41);
              \draw[line width=2pt] (52) -- (44);
              \draw[line width=2pt] (53) -- (42);
              \draw (53) -- (43);
              \draw[line width=2pt] (43) -- (31);
              \draw (44) -- (32);
              
          \end{tikzpicture}
          \caption{An instance for the case~$\alpha=3$,
            $\beta=5$ (left). The submatrix framed by the thick
            rectangle defines a bipartite graph (right). An optimal solution is
            highlighted in gray. Note that the columns in the solution
            correspond to a matching in the bipartite graph that
            saturates~$I_4$, represented by the thick lines.}
          \label{fig:iii-example}
        \end{figure}
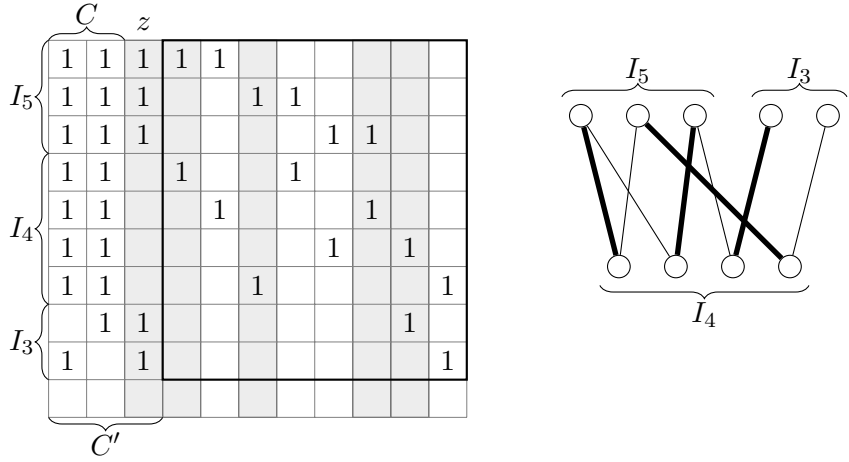%

        For $\alpha \geq 3$, we show that the matrix~$S$---recall that
        it is reduced with respect to \autoref{rr:padding}---has a
        specific structure, depicted in
        \autoref{fig:iii-example}. Namely, we claim that

        \begin{inparadesc}

        \item[(a)] if~$W_{ij} \setminus C' \neq \emptyset$ for~$i \neq j$,
          then~$i \in I_{\alpha + 1}$ and $j \in I_{\alpha} \cup
          I_{\alpha + 2}$, and
          
        \item[(b)] the unique column vector~$\col{s}{z}$ with~$z \in C \setminus C'$ equals~$\mbf{1}^n_{I_\alpha\cup I_{\alpha+2}}$.

        \end{inparadesc}
        \noindent Claim (a) implies that each column in~$[d] \setminus C'$ contains at most two 1's (naturally, any column contains at least one 1). We will see that all columns in~$[d] \setminus C'$ contain exactly two 1's and, hence, that they define the edges of a bipartite graph with the two partite vertex sets~$I_{\alpha + 1}$ and~$I_{\alpha} \cup I_{\alpha + 2}$. We find a matching that saturates~$I_{\alpha + 1}$ in this bipartite graph and show that the columns corresponding to the matching edges along with column~$z$ are an optimal solution.

\medskip
        
        To show Claim (a), observe that, if~$i \neq j \in I_{\alpha + 2}$, then~$W_{ij} \setminus C' = \emptyset$ as~$\setsys{\alpha + 2}$ is a sunflower with core~$C'$. Likewise, if~$i \neq j \in I_{\alpha + 1}$, then~$W_{ij} \setminus C' = \emptyset$ because~$\setsys{\alpha + 1}$ is a sunflower with core~$C$ and~$C \subset C'$. It hence suffices to show that~$W_{ij} \setminus C' = \emptyset$ in the case that either both~$i, j \in I_{\alpha}$ or $i \in I_{\alpha}$ and $j \in I_{\alpha + 2}$. To see the latter, note that Property~(\ref{eqn:iiiaa2})
        and \autoref{lem:sunflower-intersect} imply~$|W_i\cap C'|=\lceil\alpha/2\rceil=|C'|-1$, for all~$i\in
        I_\alpha$, that is, we even have $W_{ij}\subset C'$ for all~$i \in I_{\alpha}$, $j\in I_{\alpha+2}$. %todo even notwendig?
        Now, it only remains to show~$W_{ij} \subseteq C'$ for~$i, j
        \in I_{\alpha}$, $i \neq j$. We derived above that~$|W_i\cap
        C'|=|W_j\cap C'| = \lceil\alpha/2\rceil=|C'|-1$.
        Thus,~$|(W_i\cap C')\cap (W_j\cap C')| = |W_{ij} \cap C'| \geq |C'|-2 =
        \lfloor\alpha/2\rfloor$.
        By Property~(\ref{eqn:iiiaa}), $|W_{ij}| \leq \lfloor \alpha / 2 \rfloor$, which implies $W_{ij}\subseteq C'$.
        Hence, $W_{ij} \setminus C' = \emptyset$,
        completing the proof of Claim (a).

        % We need another condition for the above-mentioned matching to exist: namely, for each~$i \in I_{\alpha + 1}$ we have~$|W_i \setminus C'| = \lceil \alpha / 2 \rceil$ and for each

        Let us next prove Claim~(b), that
        is, $\col{s}{z}=\mbf{1}^n_{I_\alpha\cup I_{\alpha+2}}$ where~$z$ is the unique column in~$C \setminus C'$. Assume
        the contrary, that is, either a row in~$I_\alpha$ has a~0 at
        entry~$z$ or a row in~$I_{\alpha + 1}$ has a~1 at
        entry~$z$. Let us first show that~$s_{iz}=0$, that is,
        $z\not\in W_i$ is impossible for a row $i \in I_{\alpha}$.
        Using~$|W_i \cap C'| = |C'| - 1$, it follows that~$C\subset W_i$.
        Then, for all~$j\in I_{\alpha+1}\cup I_{\alpha+2}$,
        it holds~$W_{ij}\setminus C=\emptyset$ since otherwise either
        Property~(\ref{eqn:iiiaa1}) or Property~(\ref{eqn:iiiaa2}) is violated.
        Let us show that $W_{ij}\setminus C = \emptyset$ also for
        all~$j\neq i\in I_\alpha$. Recall that $W_{ij}\setminus C' = \emptyset$, as shown above. By
        assumption~$z\not\in W_{ij}$, yielding~$W_{ij}\setminus C =
        W_{ij}\setminus C' = \emptyset$. 
        But then, the columns
        in~$W_i\setminus C$ 
        equal~$\mbf{1}^n_{\{i\}}$. Note that~$|W_i \setminus C| \geq
        \alpha - \lfloor \alpha / 2 \rfloor \geq 2$ (recall
        that~$\alpha \geq 3$). Hence, there is at least one inessential column, a contradiction.
        By the same arguments (using Properties~(\ref{eqn:iiia1a1}) and~(\ref{eqn:iiia1a2})), we can infer that there is no~$i\in
        I_{\alpha+1}$ such that~$z\in W_i$, that is, $C'\subset W_i$.
        Hence, for~$z\in C'\setminus C$, it
        holds~$\col{s}{z}=\mbf{1}^n_{I_\alpha\cup I_{\alpha+2}}$, proving Claim~(b). Note that column~$z$ distinguishes all rows in~$I_{\alpha+1}$
        from all rows in~$I_\alpha\cup I_{\alpha+2}$.

        \medskip
        
        To finish the proof of Case~\ref{case011}, we need one more observation about the rows in~$I_\alpha$,
        namely that~$n_\alpha = \lceil \alpha / 2 \rceil$.
        Assume the contrary, that is, since we have~$|W_i \cap C'| =
        |C'| - 1 = |C| = \lceil\alpha/2\rceil$ and~$z\in W_i$ for all~$i \in
        I_\alpha$, there exists
        an~$x\in C$ such that~$x\in W_i$ for all~$i\in I_\alpha$.
        Then, $\col{s}{x}=\mbf{1}^n_{[n-1]}$ and thus, column~$x$
        is inessential, which is not possible.
        Hence, for each~$x\in C$, there exists an~$i\in I_\alpha$
        such that~$x\not\in W_i$.
        Since~$x\in C'$ and~$|W_i\cap C'|=|C'|-1$ it follows
        that~$W_i\cap C'= C'\setminus\{x\}$. Therefore, we
        have~$n_\alpha = |C| = \lceil \alpha / 2 \rceil$.

        We now derive a solution from the abovementioned bipartite graph. Consider the columns in~$[d]\setminus C'$.
        Clearly, if one of these columns contains only one~1,
        then this column is inessential, which yields a contradiction.
        Thus, each column contains at least two 1's.
        Using Claim~(a), each of the columns also has at most two 1's.         Also, after preprocessing, no two columns are equal.
        Thus, the submatrix~$S[[n-1],[d]\setminus C']$ (framed by thick
        lines in \autoref{fig:iii-example}) is
        the incidence matrix of a bipartite graph~$G$, where the rows
        correspond to the vertices (partitioned into~$I_{\alpha+1}$
        and~$I_{\alpha}\cup I_{\alpha+2}$) and the columns define the edges.
        Moreover, each vertex~$i \in I_{\alpha+2}$ has
        degree~$|W_i\setminus C'|=\lceil\alpha/2\rceil$, since~$\col{s}{z}=\mbf{1}^n_{I_\alpha\cup I_{\alpha+2}}$ also each vertex~$i \in I_{\alpha + 1}$ has degree~$\lceil \alpha / 2 \rceil$, and, since each row~$i \in I_\alpha$ has~$|W_i \cap C'| = \lceil \alpha / 2 \rceil$ (as derived above), each
        vertex~$i\in I_{\alpha}$ has degree~$|W_i\setminus
        C'|=\lfloor\alpha/2\rfloor$ in~$G$.
        % Since~$G$ is bipartite, we
        % have~$n_{\alpha+2}\lceil\alpha/2\rceil +
        % n_\alpha\lfloor\alpha/2\rfloor =
        % n_{\alpha+1}\lceil\alpha/2\rceil$, from which we can infer
        % $n_{\alpha+1}=n_{\alpha+2}+\lfloor\alpha/2\rfloor =
        % |I_\alpha\cup I_{\alpha+2}|-1$.
        We can now use Hall's theorem~\cite{BG09},
        to show that there exists a matching in~$G$ that saturates~$I_{\alpha+1}$,
        that is, a subset~$M\subseteq [d]\setminus C'$
        of~$n_{\alpha+1}$ columns such that~$|W_i\cap M|=1$ for all
        $i\in I_{\alpha+1}$ and~$|W_i\cap M|\le 1$ for all~$i\in
        I_\alpha\cup I_{\alpha+2}$.\footnote{Hall's theorem asserts that, for a bipartite
          graph~$G=(X\cup Y, E)$, there exists an $X$-saturating matching
          if and only if~$|T|\le |N_G(T)|$ holds for each
          subset~$T\subseteq X$.} Indeed, taking any subset~$T
        \subseteq I_{\alpha + 1}$ of vertices, consider the set~$N_G(T) \subseteq I_{\alpha} \cup I_{\alpha + 2}$ of neighbors of~$T$. Since the vertices in~$N_G(T)$ have at most the degree of any vertex in~$T$, we have~$|N_G(T)| \geq |T|$. Hence, the precondition of Hall's theorem is satisfied. Thus, $M$ exists as claimed. 

        We now claim that~$K:=M\cup \{z\}$
        with~$|K|=n_{\alpha+1}+1$ is an optimal solution
        (highlighted in gray in \autoref{fig:iii-example}).
        First, regarding $K$ being a solution, 
        since~$G$ is bipartite, we
        have~$n_{\alpha+2}\lceil\alpha/2\rceil +
        n_\alpha\lfloor\alpha/2\rfloor =
        n_{\alpha+1}\lceil\alpha/2\rceil$. From this, we can infer
        $n_{\alpha+1}=n_{\alpha+2}+\lfloor\alpha/2\rfloor =
        |I_\alpha\cup I_{\alpha+2}|-1$. Thus, as~$M$ saturates~$I_{\alpha + 1}$, there exists exactly one~$j \in I_\alpha \cup I_{\alpha + 2}$ such that $W_j \cap M = \emptyset$. Using this, it is not hard to check that~$K$ is a solution.

        Regarding optimality, it remains to show that there is no
        solution of size~$n_{\alpha+1}$. This can be seen as follows:
        \autoref{lem:sunflower-solution} implies that any solution~$K$
        intersects at least~$n_{\alpha+1}-1$ of the petals
        of~$\setsys{\alpha+1}$. If~$K$ intersects
        each petal, then, as~$n_{\alpha + 1} = n_\alpha + n_{\alpha + 2} - 1$, there exists a~$j\in I_{\alpha}\cup
        I_{\alpha+2}$ such that~$K\cap W_j=\emptyset$, which is not
        possible. Otherwise, if~$K$ intersects exactly all but one of
        the petals, then there exists an~$i\in I_{\alpha+1}$ and also~$j\neq
        j'\in I_\alpha\cup I_{\alpha+2}$ such that~$K\cap W_i\setminus
        C'=\emptyset$, $K\cap W_j\setminus C'=\emptyset$ and~$K\cap
        W_{j'}\setminus C'=\emptyset$. In order to distinguish row~$i$
        from the null vector, $K$ has to contain a column from~$C$.
        But it is not possible to pairwise distinguish all three rows~$i$, $j$,
        and~$j'$ from each other with just one column.
        Hence, $K$ is indeed optimal and~$I$ is a yes-instance if and
        only if~$k\ge n_{\alpha+1}+1$.
    \end{inparaenum}

    As regards the running time, observe that the maximum number of
    candidate solutions we have to test in any of the above cases is
    in~$O(n^2)$. Checking whether a subset of column indices is a solution
    can be done in~$O(nd)$ time via lexicographical sorting of the
    rows. This yields an overall running time in~$O(n^3d)$ which also
    subsumes the~$O(\min\{n,d\}\cdot nd)$ time for the preprocessing.
  \end{proof} 
  
  \section{Distinct Vectors on General Matrices}
  \label{sec:DVgeneral}
  In the last section, we have seen, among other results, that
  \DV{} is \NP-complete and \W{1}-hard with respect to the number~$t$
  of columns to be deleted even if the input alphabet is binary and
  the pairwise Hamming distance of the row vectors is bounded by four
  (\autoref{cor:distinct_vectors_w1-hard_t}).
  Note, however, that the parameterized complexity with respect to the
  number~$k$ of retained columns for binary alphabets remained open.
  In this section, we first show that \HS{} parameterized by the solution
  size (which is \W{2}-complete~\cite{downey13}) is
  parameterized reducible to \DV{} for
  alphabets of unbounded size, showing that \DV{} is \W{2}-hard with
  respect to~$k$ (\autoref{thm:distinct_vectors_w2-hard}).
  Nevertheless, we show later in this section some
  tractability results even for larger alphabets.
  For example, we give a problem kernel with respect to the combined
  parameter alphabet size~$|\Sigma|$ and number~$k$ of retained
  columns~(\autoref{thm:dv_kernel_sk}).
  Note that this result implies
  that \DV{} is fixed-parameter tractable with respect to~$k$
  for any alphabet of constant size.
  
  \begin{theorem}
    \label{thm:distinct_vectors_w2-hard}
    \DV{} is $W[2]$-hard with respect to the number~$k$ of retained columns. 
  \end{theorem}
  \begin{proof}
    We give a parameterized reduction from the \W{2}-complete \HS{} problem parameterized by solution size~$k$.
    \problemdef{\HS}
    {A finite universe $U$, a collection $\mathcal{C}$ of subsets of $U$, and a nonnegative integer $k$.}
    {Is there a subset $K\subseteq U$ with $|K|\leq k$ such that $K$ contains at least one element from each subset in~$\mathcal{C}$?}
    Given an instance~$(U,\mathcal{C},k)$ of \HS{} with $U=\{u_1,\dots,u_m\}$ and $\mathcal{C} = \{C_1,\dots,C_n\}$,
    we define the \DV{} instance $(S,k')$ where~$k':=k$ and the $(n+1)\times m$ matrix~$S$ is defined as
    \[s_{ij} := \begin{cases}i, & u_j\in C_i\\0, & u_j\not\in
      C_i\end{cases}\]
    for all~$i\in[n]$, $j\in[m]$ and~$\row{s}{n+1}:=\mbf{0}$.
    This instance is polynomial-time computable.
    An example is depicted in \autoref{fig:HS-DV-example}.
    If~$K\subseteq U$ is a solution of~$(U, \mathcal{C}, k)$, then~$K
    \cap C_i\neq\emptyset$ holds for all~$C_i\in\mathcal{C}$,
    and thus, for each row~$\row{s}{i}$, there exists a column~$j$
    corresponding to some element~$u_j\in K$ such that~$s_{ij}=i$.
    Since no other row contains an entry equal to~$i$, it follows
    that row~$\row{s}{i}$ is distinct from all other rows in~$S$.
    Conversely, in order to distinguish row~$\row{s}{i}$
    from~$\row{s}{(n+1)}=\mbf{0}$, any solution~$K'$ of~$(S,k')$
    has to contain a column index~$j$ such that~$s_{ij}\neq 0$.
    This implies that the subset~$\{u_j\mid j\in K'\}\subseteq U$
    contains at least one element of each~$C_i$ and is thus a solution of the original instance.
    Finally, note that this is a parameterized reduction since~$k'=k$.
  \end{proof}

  \begin{figure}[t]
    \centering
    \begin{tikzpicture}[scale=.4]
        \node[right] at (-10,6) {$U=\{1,\ldots,6\}$};
        \node[right] at (-10,4.5) {$C_1=\{1,2,3\}$};
        \node[right] at (-10,3.5) {$C_2=\{3,4\}$};
        \node[right] at (-10,2.5) {$C_3=\{1,3,6\}$};
        \node[right] at (-10,1.5) {$C_4=\{1,2,4,5\}$};
        \node[right] at (-10,0.5) {$C_5=\{1,5,6\}$};
        
        \draw[help lines] (0,0) grid (6,5);
        \foreach \y in {1,...,5}{
          \node at (-0.7,5.5-\y) {$C_\y$};
        }
        \foreach \x in {1,...,6}{
          \node at (\x-0.5,5.5) {\sffamily\scriptsize\x};
        }
        \node at (0.5,4.5) {1};
        \node at (1.5,4.5) {1};
        \node at (2.5,4.5) {1};
        \node at (2.5,3.5) {2};
        \node at (3.5,3.5) {2};
        \node at (0.5,2.5) {3};
        \node at (2.5,2.5) {3};
        \node at (5.5,2.5) {3};
        \node at (0.5,1.5) {4};
        \node at (1.5,1.5) {4};
        \node at (3.5,1.5) {4};
        \node at (4.5,1.5) {4};
        \node at (0.5,0.5) {5};
        \node at (5.5,0.5) {5};
        \node at (4.5,0.5) {5};
        \draw[thick] (2,0) rectangle (3,6);
        \draw[thick] (4,0) rectangle (5,6);
      \end{tikzpicture}
    \caption{Example of a \HS{} instance~(left) and the constructed
      matrix~(right).
    The hitting set~$K=\{3,5\}$ is indicated by thick lines.}
    \label{fig:HS-DV-example}
  \end{figure}
  
  \citet{chen05} showed that \HS{} cannot be solved in
  $|U|^{o(k)} \cdot |I|^{O(1)}$ time, unless $\FPT=\W{1}$.
  Since the reduction from \HS{} yields an instance with~$d=|U|$
  columns and solution size~$k$ in polynomial time,
  the following corollary is immediate.
  
  \begin{corollary}
    \label{cor:distinct_vectors_runningtime_lower_bound}
    If $\FPT\neq\W{1}$, then \DV{} cannot be solved in~$d^{o(k)} \cdot |I|^{O(1)}$ time.
  \end{corollary}
  \noindent
  On the positive side, \DV{} can trivially be solved by trying
  all subsets of column indices of size~$k$
  within~$d^k \cdot |I|^{O(1)}$ time.
  
  Although \autoref{thm:distinct_vectors_w2-hard} shows that
  \DV{} is \W{2}-hard with respect to the parameter~$k$,
  we can provide a problem kernel for \DV{}
  if we additionally consider the input alphabet size~$|\Sigma|$ as a second parameter.
  The size of the problem kernel is superexponential in the combined
  parameter~$(|\Sigma|,k)$.
  Clearly, a problem kernel of polynomial size would be desirable.
  However, polynomial-size problem kernels do not exist even with the additional
  parameter number~$n$ of rows, unless $\NP \subseteq \coNPpoly$,
  which would imply a collapse of the polynomial hierarchy in complexity
  theory, which is widely believed not to be the case.
  \begin{theorem}
    \label{thm:dv_kernel_sk}
    For \DV,
    \begin{enumerate}[1.)]
    \item\label{thm:dv_kernel_sk_i}
      there exists an
      $O\large({{|\Sigma|}^{{|\Sigma|}^k+k}}/{{|\Sigma|}!}\cdot\log{|\Sigma|}\large)$-size
      problem kernel computable in~$O(d^2n^2)$ time and

    \item\label{thm:dv_kernel_sk_ii}  unless $\NP \subseteq \coNPpoly$,
      there is no polynomial-size problem kernel with respect to the
      combined parameter $(n, {|\Sigma|}, k)$.
    \end{enumerate}
  \end{theorem}
  \begin{proof}
    \ref{thm:dv_kernel_sk_i}.)
    Since~$S$ contains~$n$ rows, it follows that at least~$k\ge\lceil\log_{|\Sigma|} n\rceil$
    columns are required to distinguish all rows,
    otherwise we simply return a trivial no-instance.
    Thus, we have $n\leq{|\Sigma|}^k$.
    Moreover, note that each column partitions the rows into at
    most~${|\Sigma|}$ non-empty subsets (all rows with identical
    values form a subset of the partition).
    We use the following simple data reduction rule:
    If column~$j$ partitions the rows \emph{finer} than column~$j'$
    (that is, each set in the partition of~$j$ is a subset of a set in
    the partition of~$j'$), then delete column~$j'$.
    This rule clearly is correct since column~$j$ distinguishes all
    pairs of rows that are distinguishable by column~$j'$.
    Exhaustive application of the above rule requires $O(d^2n^2)$
    arithmetic operations (checking for all~$j,j'\in[d]$
    whether~$s_{ij'}\neq s_{i'j'} \Rightarrow s_{ij} \neq s_{i'j}$
    holds for all~$i,i'\in[n]$).
    It follows that for each remaining pair of columns,
    the partition of one is not finer than the partition of the other.
    Thus, we can bound the number~$d$ of columns from above
    by~${{|\Sigma|}^n}/{{|\Sigma|}!}$. (More precisely, $d$ is bounded by the
    cardinality of a maximum \emph{antichain},
    that is, a set of partitions being pairwise incomparable with
    respect to the ``finer than'' order,
    in the \emph{partition lattice} of an~$n$-element set up to
    the~${|\Sigma|}$-th level, see \citet[Chapter IV.4]{graetzer03} for details).
    The overall size of~$S$ is thus in
    \[O(nd\cdot\log|\Sigma|)=O({{|\Sigma|}^{{|\Sigma|}^k+k}}/{{|\Sigma|}!}
    \cdot\log{|\Sigma|}),\] which yields a problem kernel with respect to the combined parameter~$(|\Sigma|,k)$.

    \ref{thm:dv_kernel_sk_ii}.)
    We give a lower bound on the size of a problem kernel based on a
    result by \citet{dom14}, who showed that there is no polynomial-size
    problem kernel for \SC{} with respect to the combined parameter~$(|U|, k)$, unless
    $\NP \subseteq \coNPpoly$ (which implies the collapse of the
    polynomial hierarchy).

    \problemdef{\SC}
    {A finite universe $U$, a collection $\mathcal{C}$ of subsets of $U$, and a nonnegative integer $k$.}
    {Is there a subset $S\subseteq \mathcal{C}$ with $|S|\leq k$ such
      that each element of~$U$ is contained in at least one subset in~$S$?}
    
    The reduction from \HS{} in the proof
    of~\autoref{thm:distinct_vectors_w2-hard} can be used to obtain a reduction from \SC{}, when first transforming \SC{} into
    \HS{} in the common way~\cite{AAP80}, that is, the universe of the \HS{}
    instance is~$\mathcal{C}$ and for each element $u\in U$, there is the
    subset~$\{C\in\mathcal{C}\mid u\in C\}$.
    The resulting \DV{} instance consists of a matrix
    with~$n = |U|$ rows over an alphabet of size~${|\Sigma|} = |U|+1$
    and a sought solution size~$k$.
    Since \SC{} is NP-complete, there is a polynomial-time many-one reduction
    from \DV{} to \SC{} and, hence, a polynomial-size kernel for \DV{}
    would imply a polynomial-size kernel for \SC{}: simply transform the
    \SC{} instance into a \DV{} instance, kernelize, and transform
    back.
  \end{proof}

  Observe the gap between the superpolynomial lower bound and the
  superexponential upper bound on the problem kernel size in
  \autoref{thm:dv_kernel_sk}, which leaves a significant gap. Proving the (non)-existence of a~$|\Sigma|^{O(k)}$-size
  problem kernel, for example, would be an interesting result.
  
  We now move on to parameterizing by the maximum pairwise row Hamming distance~$H$.
  Recall \autoref{def:hamming}, where, for a
  matrix~$S\in\Sigma^{n\times d}$,
  we defined~$H:=\max_{i\neq j\in[n]}\Ham(\row{s}{i},\row{s}{j})$.
  In this case, every pair of rows in~$S$ differs in at most~$H$
  columns, which yields a kernelization and also a fairly simple
  approximation algorithm by a reduction from \DV{} to \HS{}:

  \begin{theorem}
    \label{thm:dv_kernel_hk}
    Let $H$ be the maximum pairwise row Hamming distance of the
    input matrix. Then, \DV{}
    \begin{enumerate}[1.)]
    \item  \label{thm:dv_kernel_hk_i} is linear-time factor-$H$ approximable and
    \item  \label{thm:dv_kernel_hk_ii} admits an $O(g(H,k)^2\log g(H,k))$-size
      problem kernel which can be computed in $O(d^2 + n^2\max\{d\log d,
      dn^2\})$ time, where~$g(H,k):=H!\cdot H^{H+1}\cdot(k+1)^H$.
    \end{enumerate}
  \end{theorem}

  \begin{proof}
    The idea for both results is to define a polynomial-time parameterized
    many-one reduction from \DV{} to \fHS{}, which is the special case of \HS{} where
    each input set has cardinality at most~$H$.
    We can then apply known kernelization and approximation algorithms
    to the \fHS{} instance.
    The reduction works as follows:
    Given an instance~$(S,k)$ of \DV{}, the \fHS{}
    instance~$(U,\mathcal{C},k')$ is defined as
    \[
      U:=[d],
      \mathcal{C}:=\{C_{ij}\subseteq U\mid i\neq j\in[n]\}, \text{
        where } C_{ij}:=\{u\in U\mid s_{iu}\neq s_{ju}\},
    \]
    and~$k':=k$. Note that $|C_{ij}|\leq H$ holds for all $i\neq j$.
    This reduction requires $O(n^2d)$ arithmetic operations.
    It is correct since~$K\subseteq[d]$ with $|K|\leq k$ is a solution of $(S,k)$
    if and only if for every pair of rows in~$S$ there is at least one column in~$K$
    in which both rows have different values.
    This is equivalent to the situation that~$K$ contains at least one element
    from each~$C_{ij}$ in~$\mathcal{C}$, which implies that~$K$ is a
    solution of~$(U,\mathcal{C},k)$.
    
    We now prove the two statements of the theorem using the above
    reduction.

    \ref{thm:dv_kernel_hk_i}.) A factor-$H$ approximation algorithm
    repeatedly adds a so far unhit subset to the hitting set.
    
    \ref{thm:dv_kernel_hk_ii}.)
    Let~$(U,\mathcal{C},k)$ with~$|U|=d$ and~$|\mathcal{C}|\in O(n^2)$
    be the \fHS{} instance resulting from the above reduction.
    We apply a \fHS{} kernelization due to \citet{bevern13} in order to obtain
    in~$O(Hd + H\log H\cdot n^2 + Hn^4)$ time an
    instance~$(U',\mathcal{C}',k)$, where~$|U'|$
    and~$|\mathcal{C}'|$ are at most~$g(H,k)$.
    In order to obtain a problem kernel for \DV{},
    we transform the instance~$(U',\mathcal{C}',k)$ back by the
    reduction from the proof of
    \autoref{thm:distinct_vectors_w2-hard}.
    We end up with a \DV{} instance~$(S',k')$ with~$k'=k$ in
    $O(|U'|\cdot|\mathcal{C}'|)=O(n^2d)$~time.
    Since~$(U',\mathcal{C}',k)$ is an instance of \fHS,
    it follows that each row in~$S'$ differs from~$\mathbf{0}$ in
    at most~$H$ columns.
    Thus, each pair of rows in~$S'$ differs in at most~$H'\le 2H$
    columns.
    Note that~$k'$ and~$H'$ depend only on~$k$ and~$h$,
    which also holds for the overall size of~$S'$, which is in
    $O(|U'|\cdot|\mathcal{C}'|\log|\mathcal{C}'|)= O(g(H,k)^2\log (g(H,k)))$.
    Moreover, the overall running time is in $O(n^2d + Hd + H\log
    H\cdot n^2 + Hn^4)$, which gives a problem kernel.
  \end{proof}
  
  In this section, we have seen that \DV{} can basically be
  regarded as a special \HS{} problem. \HS{} in
  general is \W{2}-complete \cite{downey13} with respect to the
  solution size,
  but \DV{} is fixed-parameter tractable with respect to the solution
  size for constant-size alphabets
  (\autoref{thm:dv_kernel_sk}).
  Thus, the set systems induced by constant-size alphabet instances of
  \DV{} involve a certain structure (that is, the number of subsets is
  exponentially upper-bounded in the size of the solution) that makes them
  somewhat easier to solve.
  Generalizing the analysis of the structure as we did in
  \autoref{sec:DVbinary} for binary alphabets to arbitrary alphabets
  deserves further investigation.

\section{Conclusion}
\label{sec:conc}
We conclude with a few challenges for future research.
Based on pairwise minimum and maximum Hamming distances, we 
proved a complexity dichotomy for \textsc{Distinct Vectors} 
restricted to binary matrices. 
We leave generalizations of the polynomial-time solvable cases to non-binary alphabets as a major open question.
A further interesting question is whether one can close the gap 
between the doubly-exponential upper and the superpolynomial lower bound 
for the size of the problem kernel for \textsc{Distinct Vectors} parameterized by
the combined parameter ``alphabet size and number of remaining columns''.

From a combinatorial point of view, the study of ``vector problems'' in general seems to be a fertile but little researched area for parameterized complexity studies.
Another example of a parameterized complexity analysis for a vector problem deals with the explanation of integer vectors by few homogenous segments~\cite{BCHKNS15}.
Finally, on a more general scale, it seems that parameterized complexity 
analysis is a promising tool to better assessing the computational
complexity of some machine learning problems such as combinatorial feature
selection; our work is among the first contributions in 
this so far widely neglected research direction.

%\todo[inline]{Very general. Rewrite. Pose concrete open questions?}
%Typically, combinatorial feature selection problems are
%computationally very hard in the worst case.  That is why practical
%solution approaches are often heuristic in nature.  Parameterized complexity
%analysis, however, is a tool that offers new theoretical perspectives
%for analyzing and better understanding their computational
%complexity. Ideally, identifying problem-specific parameters helps in
%spotting new islands of provable (fixed-parameter) tractability that
%may be a theoretically substantiated and practically viable
%alternative to purely heuristic approaches.
%Many of our results indicate computational hardness also
%for very restricted (parameterized) problem instances, but we could
%also spot some parameter constellations that might help to efficiently
%find exact solutions.

\bibliographystyle{abbrvnat}
\bibliography{ref}
\end{document}